\documentclass[11pt,reqno]{article}
\usepackage[margin=1in]{geometry}
\usepackage{dsfont, amssymb,amsmath,amscd,latexsym, amsthm, amsxtra,amsfonts}
\usepackage{lineno}
\usepackage[all]{xy}
\usepackage[active]{srcltx}
\usepackage{tikz}
\usepackage[round]{natbib}
\usepackage{bbm}
\usepackage{enumerate}
\usepackage{mathrsfs}
\usepackage{graphicx}
\usepackage{subcaption}
\usepackage{comment}
\usepackage{mathtools}
\usepackage{cases}
\usepackage{tcolorbox}
\tcbuselibrary{most}

\usetikzlibrary{calc,arrows}
\usepackage{verbatim}
\usepackage{color}
\usepackage[affil-it]{authblk}
\usepackage{bm}
\usepackage[title]{appendix}

\newtheorem{theorem}{Theorem}[section]
\newtheorem{assumption}[theorem]{Assumption}

\newtheorem{definition}[theorem]{Definition}

\newtheorem{lemma}[theorem]{Lemma}

\newtheorem{proposition}[theorem]{Proposition}
\newtheorem{remark}[theorem]{Remark}
\numberwithin{equation}{section}

\DeclareMathOperator*{\esssup}{ess\,sup\,}
\newcommand{\logn}{\mathrm{LogNormal}}
\newcommand{\ppm}{{\prime\prime}}

\newcommand{\md}{\mathrm{d}}

\newcommand{\mR}{\mathbb{R}}

\newcommand{\mE}{\mathbb{E}}
\newcommand{\mF}{\mathbb{F}}
\newcommand{\mP}{\mathbb{P}}

\newcommand{\ind}{\mathbbm{1}}

\renewcommand{\epsilon}{\varepsilon}

\newcommand{\F}{\mathcal{F}}

\newcommand{\Ti}{\mathcal{T}}

\newcommand{\cR}{\mathcal{R}}

\newcommand{\barpi}{\bar{\pi}}

 \usepackage[pdfstartview=FitH, bookmarksnumbered=true,bookmarksopen=true, colorlinks=true, pdfborder={0 0 1}, citecolor=blue, linkcolor=blue,urlcolor=blue]{hyperref}
\usepackage{graphics}
\graphicspath{{figures/}}

\title{Dynamic portfolio selection under generalized disappointment aversion}

\author{Zongxia Liang$^{a,e}$\thanks{Email: \texttt{liangzongxia@tsinghua.edu.cn}}\ \ \ \ \ Sheng Wang$^b$\thanks{Email: \texttt{sheng-wa15@tsinghua.org.cn}} \ \ \ \ \ Jianming Xia$^c$\thanks{Email: \texttt{xia@amss.ac.cn}} \ \ \ \ \ Fengyi Yuan$^d$\thanks{Email: \texttt{yuanfengyi@cuhk.edu.cn}}
}	
\affil{$^a$Department of Mathematical Sciences, Tsinghua University, China \\$^b$Department of Statistics and Data Science, The Chinese University of Hong Kong, Shatin, N.T., Hong
Kong SAR\\
$^c$State Key Laboratory of Mathematical Sciences, Academy of Mathematics and Systems Science, Chinese Academy of Sciences, Beijing,
China\\
$^d$ School of Science and Engineering, the Chinese University of Hong Kong (Shenzhen), Shenzhen, China\\
$^e$  SEM,  Tsinghua University, China       }

\numberwithin{equation}{section}

\begin{document}
		\maketitle

\begin{abstract}
This paper addresses the continuous-time portfolio selection problem under generalized disappointment aversion (GDA). The implicit definition of the certainty equivalent within GDA preferences introduces time inconsistency to this problem.  Within the class of bounded, right-continuous deterministic strategies, we provide a necessary and sufficient condition for equilibrium in the form of a fully nonlinear integral equation. Investigating the existence and uniqueness of the solution to the integral equation, we establish the existence and uniqueness of the equilibrium  in this class. Our findings indicate that under disappointment aversion preferences, non-participation in the stock market is the unique equilibrium  in this class. The semi-analytical equilibrium strategies obtained under the constant relative risk aversion utility functions reveal that, under GDA preferences, the investment proportion in the stock market consistently remains smaller than the investment proportion under classical expected utility theory. The numerical analysis shows that the equilibrium strategy's monotonicity concerning the two parameters of GDA preference  aligns with the monotonicity of the degree of risk aversion. Finally, in a switching-drift market with a random expected return, we prove the existence of a non-deterministic equilibrium. 
\end{abstract}
		

\section{Introduction}{\label{intro}}
\noindent
Since the seminal work of \cite{Merton1971}, portfolio selection within the expected utility (EU) framework has been dominating in modern financial theory. However, various empirical and experimental studies show that EU is unrealistic, as exemplified by the famous Allais paradox \citep{Allais1953}. This motivates  the emergence of various alternative models on preferences, including the disappointment aversion (DA) preference of \cite{Gul1991} and the generalized disappointment aversion (GDA) preference of \cite{Routledge2010}.  

GDA (including DA as a special case) preferences have numerous applications in asset pricing and portfolio selection. \cite{Routledge2010}, \cite*{bonomo2011generalized}, \cite{Liu2015Miao},  \cite{schreindorfer2020macroeconomic},  \cite{augustin2021disappointment}, and \cite{babiak_2023} explore the applications of GDA preferences in consumption-based asset pricing with discrete-time recursive utility models. 
\cite*{ang2005stocks}, \cite{fielding2007myopic}, \cite{Saltari2010}, \cite*{dahlquist2017asymmetries}, \cite{Ferland2021},  and \cite*{kontosakos2024long} investigate the problem of portfolio selection to maximize the GDA  certainty equivalent of the terminal wealth within discrete time models (including single-period models), while \cite{yoon2009optimal} investigates the problem with a continuous time model.

Despite their popularity, it is difficult to investigate the dynamic portfolio selection problem for GDA preferences, due to the implicit definition of the  certainty equivalent, which leads to the time inconsistency of the preferences. \cite{yoon2009optimal}  seeks  the optimal solution, also known as the pre-commitment solution,  for DA preferences, regardless of the time inconsistency: the investor may violate the current strategy in the future. \cite*{ang2005stocks} first examine  the dynamic portfolio selection for DA preferences in a discrete time setting, employing the backward induction to get time-consistent solutions. However, the number of states increases exponentially with the number of periods, and therefore the backward induction is difficult to implement.  Then they consider a 
reduced model in which the current period's  certainty equivalent relies on the next period's  certainty equivalent only, instead of the actual future returns. This reduced model addresses the issue of exponentially growing number of states. However, it deviates from the original problem: the recursive definition of the  certainty equivalent in this reduced model results in that,   only in the next-to-last period, does the  certainty equivalent coincide with the original one of \cite{Gul1991}. Recently, \cite*{kontosakos2024long} apply this reduced model  to investigate the influence of return predictability and parameter uncertainty on dynamic portfolio selection under DA preferences.

In this paper, we investigate the portfolio selection problem for GDA preferences in continuous time.  As a continuous-time version of the problem presented in \cite*{ang2005stocks}  and \cite*{kontosakos2024long}, we similarly focus solely on investment strategies, excluding consideration of consumption strategies.   In contrast to \cite{yoon2009optimal}, we focus on seeking a time consistent solution, following the ``consistent planning" (intra-personal equilibrium) approach of \cite{Strotz1955}.  Nevertheless, backward induction proves ineffective in  continuous time, as there is no specific time point closest to the terminal time. The issue of time inconsistency in continuous time is therefore challenging and progresses slowly. The precise definition of continuous-time intra-personal equilibrium was first introduced by 
\cite{Ekeland2010} (a previous version of this paper: \cite{ekeland2006being}), when they addressed  the time inconsistency arising from non-exponential discounting. Then numerous studies have explored various problems with time inconsistency in continuous-time. Among others, \cite*{Bjork2017}, \cite{he2021equilibrium}, and \cite{Hernandez2023} investigate general stochastic control problems with time-inconsistent cost functions; \cite{basak2010dynamic}, \cite*{hu2012time,Hu2017}, \cite*{bjork2014mean}, and \cite*{Dai2021}
explore the dynamic mean-variance portfolio selection or linear-quadratic stochastic control problems;
\cite{ekeland2008investment} and \cite{hamaguchi2021time}  discuss portfolio selection problems involving non-exponential discounting; \cite*{Hu2021} examine the continuous-time portfolio selection for rank-dependent utilities; Recently, \cite*{liang2023dynamic} and \cite*{liang2023equilibrium} study the continuous-time stochastic control and portfolio selection with implicitly defined objective functions. 

Given a utility function $U$ and parameters $\beta\ge 0$ and $\delta>0$, the GDA value $\eta(Y)$ of an outcome $Y$ is implicitly determined by
\begin{equation*}
	U(\eta(Y))=\mE\left[U(Y)\right]-\beta \mE\left[(U(\delta \eta(Y))-U(Y))_+\right],
\end{equation*}
where $x_+=\max\{x,0\}$ for $x\in\mR$. 
In the previous equation,  $\delta\eta(Y)$ serves as the benchmark, where $\delta$ is the adjustment coefficient between the benchmark level and the GDA value.
The penalty term 
$\beta\mE\left[\left(U(\delta \eta(Y))-U(Y)\right)_+\right]$
captures the agent’s aversion to being disappointed by the outcome values below the benchmark level.
The parameter $\beta$ measures how disappointment averse the agent is, whereas the parameter $\delta$ measures how easily the agent becomes disappointed. 
A DA preference is a GDA preference with $\delta=1$ and an EU preference is a GDA preference with $\beta=0$. 

To compare the risk aversion between two GDA preferences, we need to use certainty equivalents. When $\delta\in(0,1]$, the GDA value and the certainty equivalent coincide. In contrast, when $\delta>1$, the GDA value of a deterministic outcome is smaller than the outcome itself, and therefore the GDA value can not stand for the certainty equivalent. To address this issue, \citet[p. 1309]{Routledge2010} redefine the certainty equivalent under the constant relative risk aversion (CRRA) utility functions (excluding the logarithmic utility function), ensuring that the certainty equivalent of a deterministic outcome is the outcome itself. However, their definition is very restrictive: it applies only to a specific range of the two parameters $\beta$ and $\delta$; see Remark \ref{rmk:rz} below. We propose an alternative definition of the certainty equivalent, which is applicable to general utility functions and  any values of the two parameters. More importantly, we prove the monotonicity of the risk aversion of the GDA preferences with respect to the two parameters, which is helpful in explaining the monotonicity of the equilibrium strategies with respect to the two parameters.

To investigate the portfolio selection problem for GDA preferences, we adopt continuous-time intra-personal equilibrium following the aforementioned literature. For a market with deterministic coefficients, it turns out that we can find  equilibrium strategies in the class of deterministic strategies. The time-$t$ GDA value of a deterministic strategy is a function $g$ of the cumulative risk and the cumulative return over time interval $(t,T]$. Then, an equilibrium  in this class results from  the trade-off between the instantaneous risk and the instantaneous return. The equilibrium strategy takes the form of the market price of risk multiplied by a scalar, which is entirely determined by the marginal rate of substitution of the cumulative risk for the cumulative return.

In the case $\delta\ne1$, we further characterize the equilibrium strategies  in this class in terms of the solutions of a class of fully non-linear integral equations and observe that equilibrium investments remain non-zero whenever the expected return of stocks is non-zero, which is consistent with the observations of \cite*{dahlquist2017asymmetries}. In the case of DA preferences, we show that  $0$ is  the unique equilibrium strategy in this class, signifying non-participation in the stock market. This holds true for any utility function $U$ and parameter $\beta>0$,  echoing the insights shared in \cite*{ang2005stocks}.

Moving on to the CRRA utility framework, we derive semi-analytical equilibrium strategies. The introduction of GDA ($\delta\neq1$) is observed to alter the agent's risk attitude, rendering the agent more risk averse.  As a result, the equilibrium strategy consistently falls below the optimal investment levels identified by \cite{Merton1971}, which cannot be obtained theoretically under the discrete-time multi-period models because  only  numerical solutions are available.
Moreover, our numerical analysis indicates that when $0 < \delta < 1$, a gradual increase in $\delta$ leads to a gradual decrease in equilibrium investment, while the opposite holds true for the case $\delta > 1$.
Furthermore, as $\beta$ gradually increases, signifying a higher aversion to disappointment, the equilibrium investment decreases. 

The existing literature lacks reporting on the continuous-time intra-personal equilibrium portfolio selection for GDA preferences. Our paper fills this gap, presenting a threefold contribution. First, we redefine the certainty equivalent for GDA preferences with  general utility functions and  a broader range of parameters when $\delta>1$. In addition, we demonstrate the monotonicity of risk aversion for GDA preferences with respect to the two parameters, $\beta$ and $\delta$. Second, under GDA preferences ($\delta\neq1$), we prove the existence and uniqueness of the solution to the integral equation arising from the trade-off between risk and  return, thereby  establishing the uniqueness and existence of the equilibrium  in the deterministic class. Furthermore, under the CRRA utility framework, we obtain semi-analytical equilibrium strategies and find that equilibrium investments are smaller than those of EU preferences.
Third,  within the context of DA preferences ($\delta=1$),  we demonstrate that  $0$ is the unique equilibrium strategy in this class, regardless of the assigned value of $\beta>0$. 

The remainder of the paper is organized as follows: Section \ref{sec:formulation} formulates the market model, GDA preferences and the portfolio selection problem. In Section \ref{sec:char:equi}, we characterize the equilibrium condition as an integral equation and  prove the existence and uniqueness of the equilibrium  within the deterministic class. In Section \ref{sec:CRRA}, we consider the CRRA utility and study the equilibrium strategies numerically. Section \ref{sec:switching-drift} provides a switching-drift example with a non-deterministic equilibrium. 
The proofs for lemmas  are collected in the Appendix.

\section{Problem formulation}\label{sec:formulation}
\noindent
In this section we formulate the financial market model and the portfolio selection problem.   

\subsection{Financial market}
\noindent
Let $T>0$ be a finite time horizon and  $\left(\Omega,\mathcal{F},\mathbb{F},\mathbb{P}\right)$ be a filtered complete probability space, where $\mathbb{F}=\left\{\mathcal{F}_{t}\right\}_{0\leq t\leq T}$ is the filtration generated by a standard $d$-dimensional Brownian motion $B=\{B(t):= \bigl(B_{1}(t),\cdots,B_{d}(t)\bigr)^{\top}, 0\leq t\leq T\}$, augmented by all null sets. 

The market consists of one risk-free asset (bank account) and $d$ risky assets (stocks).  
For simplicity, we assume that the interest rate of the bank account is zero.
The stock price processes $S_{i}$, $i=1,\cdots,d$, follow the dynamics
\begin{equation*}
	dS_{i}(t)=S_{i}(t)\left[\mu_{i}(t)\md t+\sigma_{i}(t) \md B(t)\right],\quad t\in[0,T],\, i=1,\cdots,d,
\end{equation*}
where the market coefficients $\mu:[0,T]\rightarrow\mathbb{R}^{d}$ and $\sigma:[0,T]\rightarrow\mathbb{R}^{d\times d}$ are bounded, right-continuous, and deterministic,  $\sigma_i$ denotes the $i$-th row of $\sigma$. Moreover, there are two positive constants $c_{1}$ and $c_{2}$ such that 
\begin{equation}\label{c}
	c_{1}\lVert\alpha\rVert^{2}\leq\lVert\sigma^{\top}(t)\alpha\rVert^{2}\leq c_{2}\lVert\alpha\rVert^{2},\quad \forall\alpha\in\mathbb{R}^{d}\,\,\text{and}\,\,t\in[0,T].
\end{equation}

\subsection{Generalized disappointment aversion preference}
\noindent
Consider a utility function $U:(0,\infty)\to\mR$, which is continuous and strictly increasing. Let $Y$ be a strictly positive random variable. 
For every $t\in[0,T)$, the time-$t$ \emph{GDA value}, denoted by $\eta_t(Y)$, of the outcome $Y$ is a strictly positive $\F_t$-measurable random variable that satisfies the following equation:
\begin{equation}\label{eq:etat(Y)}
	U(\eta_t(Y))=\mE_t\left[U(Y)\right]-\beta \mE_t\left[(U(\delta \eta_t(Y))-U(Y))_+\right],
\end{equation}
where $\beta\ge0$  and $\delta>0$ are preference parameters,  $\mE_t$ is the conditional expectation given $\F_t$, and $x_+=\max\{x,0\}$ for $x\in\mR$. 
The next lemma shows that the GDA value $\eta_t(Y)$ is well defined if $\mE_t[|U(Y)|]<\infty$ a.s. 

\begin{lemma}\label{lma:eta}
	Suppose that $U:(0,\infty)\to\mR$ is continuous and strictly increasing. Let $t\in[0,T)$ be fixed and $Y$ be a strictly positive random variable with $\mE_t[|U(Y)|]<\infty$ a.s.
	Then there exists a unique $\mathcal{F}_t$-measurable, strictly positive random variable $\eta$ such that
	$$U(\eta)=\mE_t\left[U(Y)\right]-\beta \mE_t\left[(U(\delta \eta)-U(Y))_+\right].$$
\end{lemma}
\begin{proof}
	See Appendix \ref{sec:proof:eta1}.
\end{proof}

In Equation \eqref{eq:etat(Y)},  $\delta\eta_t(Y)$ serves as the benchmark, where $\delta$ is the adjustment coefficient between the benchmark level and the GDA value.
Once the outcome $Y$ is below the benchmark level $\delta\eta_t(Y)$, there is a shortfall $U(\delta \eta_t(Y))-U(Y)$ of utility. The  term 
$\beta\mE_t\left[\left(U(\delta \eta_t(Y))-U(Y)\right)_+\right]$
is a penalty in the calculation of $\eta_t(Y)$, which captures the agent’s disappointment aversion to the shortfall.
The parameter $\beta$ measures how disappointment averse the agent is, whereas the parameter $\delta$ measures how easily the agent becomes disappointed. 
Such a preference is called a \emph{generalized disappointment aversion} (GDA) preference; see \cite{Routledge2010}.  

\begin{remark}\label{rmk:gda} \
	In the case $\beta=0$, the agent is disappointment neutral and $\eta_t(Y)= U^{-1}\left(\mE_t\left[U(Y)\right]\right)$, which represents the classical EU preference. 
	In the case $\beta>0$ and $\delta=1$, the preference reduces to the \emph{disappointment aversion} (DA) preference of \cite{Gul1991}.
\end{remark}

The time-$t$ \emph{certainty equivalent} of $Y$ is the strictly positive and $\F_t$-measurable random variable $C_t(Y)$ which is indifferent to $Y$:
\begin{equation}\label{eq:Ct(Y)}
	\eta_t\left(C_t(Y)\right)=\eta_t(Y).
\end{equation}

Assume $\delta\in(0,1]$. Then
\begin{equation}\label{eq:eta(Z)}
	\eta_t(Z)=Z\text{ if $Z$ is strictly positive and $\F_t$-measurable,}
\end{equation} 
which implies that $C_t(Y)=\eta_t(Y)$ if $\mE_t[|U(Y)|]<\infty$ a.s. 

In the case $\delta>1$, however,
\eqref{eq:eta(Z)} is not true any more. In fact, we have $\eta_t(Z)=\psi(Z)<Z$ if $Z$ is strictly positive and $\F_t$-measurable,
where function $\psi:(0,\infty)\to(0,\infty)$ is defined implicitly by
$$U(\psi(w))+\beta U(\delta\psi(w))=(1+\beta)U(w),\quad w\in(0,\infty).$$
Therefore, $\eta_t(Y)$ is not the time-$t$ certainty equivalent of $Y$. In this case, 
as both $C_t(Y)$ and $\eta_t(Y)$ are $\mathcal{F}_t$-measurable, it is straightforward to see that \eqref{eq:Ct(Y)} is equivalent to 
\begin{align}\label{eq:C:compact}
	U(\eta_t(Y)) = U(C_t(Y)) - \beta\Big(U(\delta \eta_t(Y)) - U(C_t(Y))\Big).
\end{align}
Therefore,
$C_t(Y)=\varphi(\eta_t(Y))$, where the function
$\varphi:(0,\infty)\to(0,\infty)$ is defined by
\begin{equation}\label{eq:varphi}
	\varphi(w)=\psi^{-1}(w)=U^{-1}\left({U(w) + \beta U(\delta w)\over 1+\beta}\right), \quad w\in(0,\infty).
\end{equation}

We will occasionally use $C_t(Y,\delta,\beta)$ to denote the time-$t$ certainty equivalent of $Y$ to highlight the dependence on the parameters $\delta$ and $\beta$. The following theorem shows the monotonicity of $C_t(Y,\delta,\beta)$ with respect to $\delta$ and $\beta$.
\begin{theorem}\label{prop:C:delta}
	Suppose that $U$ and $Y$ satisfy the conditions in Lemma \ref{lma:eta}.  Then we have the following three assertions:
	\begin{enumerate}
		\item[(i)]For any $\beta\in(0,\infty)$ and $0<\delta_1<\delta_2\leq1$, we have
		$C_t(Y,\delta_1,\beta)\ge C_t(Y,\delta_2,\beta)$;
		\item[(ii)] For any $\beta\in(0,\infty)$ and $1\le\delta_1<\delta_2<\infty$, we have
		$C_t(Y,\delta_1,\beta)\leq C_t(Y,\delta_2,\beta)$;
		\item[(iii)] For any $\delta\in(0,\infty)$ and  $0<\beta_1<\beta_2<\infty$, we have
		$C_t(Y,\delta,\beta_1)\geq C_t(Y,\delta,\beta_2)$.
	\end{enumerate}
\end{theorem}		
\textbf{Proof.} 	We use $\eta_t(Y,\delta,\beta)$ to denote the time-$t$ GDA value of $Y$ to highlight the dependence on the parameters. 
	
	Assertion (i) is obvious, because $C_t(Y,\delta,\beta)=\eta_t(Y,\delta,\beta)$ for $\delta\in(0,1]$. 
	
	We now prove assertion (ii). For simplicity of notation, we drop  $Y$ and $\beta$ in $C_t(Y,\delta,\beta)$ and $\eta_t(Y,\delta,\beta)$. Assume $1\le\delta_1<\delta_2<\infty$.   Then
	\begin{align*}
		&C_t(\delta_1)\leq C_t(\delta_2)\\\Leftrightarrow& U(\eta_t(\delta_1)) + \beta U(\delta_1\eta_t(\delta_1))\leq U(\eta_t(\delta_2)) + \beta U(\delta_2\eta_t(\delta_2))\\
		\Leftrightarrow& - \mE_t\left[(U(\delta_1 \eta_t(\delta_1))-U(Y))_+\right]+U(\delta_1\eta_t(\delta_1))\leq- \mE_t\left[(U(\delta_2 \eta_t(\delta_2))-U(Y))_+\right]+U(\delta_2\eta_t(\delta_2))
		\\
		\Leftrightarrow&\mE_t\left[\min\left\{U(Y),U(\delta_1\eta_t(\delta_1))\right\}\right]\leq \mE_t\left[\min\left\{U(Y),U(\delta_2\eta_t(\delta_2))\right\}\right],
	\end{align*}
	where the first "$\Leftrightarrow$" uses \eqref{eq:C:compact}, the second "$\Leftrightarrow$"   uses the definition of $\eta_t$, and 
	the third "$\Leftrightarrow$" uses the fact that $-(a-b)_++a=\min\{a,b\}$. Therefore, it suffices to show 
	$\delta_1\eta_t(\delta_1)\leq \delta_2\eta_t(\delta_2)$.
	By the definition of $\eta_t$, we have 
	$$U\left(\frac{\delta_1\eta_t(\delta_1)}{\delta_1}\right)+\beta \mE_t\left[(U(\delta_1 \eta_t(\delta_1))-U(Y))_+\right]=		U\left(\frac{\delta_2\eta_t(\delta_2)}{\delta_2}\right)+\beta \mE_t\left[(U(\delta_2 \eta_t(\delta_2))-U(Y))_+\right],$$
	which obviously implies $\delta_1\eta_t(\delta_1)\leq \delta_2\eta_t(\delta_2)$. Thus, assertion (ii) is proved.

	Now we prove assertion (iii). The proof in the case $\delta\in(0,1]$ is trivial because $C_t(Y,\delta,\beta)=\eta_t(Y,\delta,\beta)$ for $\delta\in(0,1]$. 
	Assume $\delta>1$. Again, we drop $Y$ and $\delta$ in $C_t(Y,\delta,\beta)$ and $\eta_t(Y,\delta,\beta)$. Assume $0<\beta_1<\beta_2<\infty$. Then
	\begin{align*}
		&C_t(\beta_1)\geq C_t(\beta_2)\\\Leftrightarrow& \frac{ \beta_1 U(\delta\eta_t(\beta_1))+U(\eta_t(\beta_1))}{1+\beta_1}\geq \frac{\beta_2 U(\delta\eta_t(\beta_2))+U(\eta_t(\beta_2))}{1+\beta_2}\\
		\Leftrightarrow&(1+\beta_2)\beta_1U(\delta\eta_t(\beta_1))+(1+\beta_2)U(\eta_t(\beta_1))\geq (1+\beta_1)\beta_2U(\delta\eta_t(\beta_2))+(1+\beta_1)U(\eta_t(\beta_2))\\
		\Leftrightarrow&(1+\beta_2)\beta_1U(\delta\eta_t(\beta_1))+(1+\beta_2)\left\{\mE_t\left[U(Y)\right]-\beta_1\mE_t\left[\left(U(\delta\eta_t(\beta_1))-U(Y)\right)_+\right]\right\}\\
		&\phantom{eeeeeeeee}\geq (1+\beta_1)\beta_2U(\delta\eta_t(\beta_2))+(1+\beta_1)\left\{\mE_t\left[U(Y)\right]-\beta_2\mE_t\left[\left(U(\delta\eta_t(\beta_2))-U(Y)\right)_+\right]\right\}\\
		\Leftrightarrow&(1+\beta_2)\beta_1\mE_t\left[\min\{U(Y),U(\delta\eta_t(\beta_1))\}-U(Y)\right]\geq(1+\beta_1)\beta_2\mE_t\left[\min\{U(Y),U(\delta\eta_t(\beta_2))\}-U(Y)\right]\\
		\Leftrightarrow&(1+\beta_2)\beta_1\mE_t\left[\left(U(Y)-U(\delta\eta_t(\beta_1) )\right)_+\right]\leq(1+\beta_1)\beta_2\mE_t\left[\left(U(Y)-U(\delta\eta_t(\beta_2) )\right)_+\right],
	\end{align*}
	where the first "$\Leftrightarrow$" uses \eqref{eq:C:compact}, the second "$\Leftrightarrow$" is obvious, the third "$\Leftrightarrow$"  uses the definition of $\eta_t$, the fourth "$\Leftrightarrow$" uses $a-(a-b)_+=\min\{a,b\}$ and $\beta_2-\beta_1=(1+\beta_1)\beta_2-(1+\beta_2)\beta_1$, and the last "$\Leftrightarrow$" uses $a-\min\{a,b\}=(a-b)_+$.
	Therefore, it is left to show 
	$\eta_t(\beta_1)\geq \eta_t(\beta_2)$ because $(1+\beta_2)\beta_1<(1+\beta_1)\beta_2$. 
	By the definition of $\eta_t$, we have 
	$$U(\eta_t(\beta_1))+\beta_1 \mE_t\left[\left(U(\delta\eta_t(\beta_1))-U(Y)\right)_+\right]=U(\eta_t(\beta_2))+\beta_2 \mE_t\left[\left(U(\delta\eta_t(\beta_2))-U(Y)\right)_+\right],$$
	which obviously implies $\eta_t(\beta_1)\geq \eta_t(\beta_2)$. Thus, assertion (iii) is proved.\hfill $\square$

By definition,  the GDA value $\eta_t(Y)$ decreases with respect to the parameter $\delta$. When $\delta\in(0,1]$, the GDA value $\eta_t(Y)$ coincides with certainty equivalent $C_t(Y)$.  Thus, the GDA preference exhibits more risk aversion as $\delta$ becomes larger in $(0,1)$.
This trend, however, does not necessarily extend to   $(1,\infty)$. 
When $\delta>1$, the GDA value $\eta(w)$ of a deterministic outcome $w$ is not the outcome $w$ itself. 
In this case,
it is not appropriate to compare risk aversion using the GDA value. 
Therefore, we transform the GDA value to the certainty equivalent through a function $\varphi$,  the inverse of the  function $\psi$. This yields the relationship  $$\text{certainty equivalent} = \varphi(\text{GDA value}).$$
From \eqref{eq:varphi}, we know that the function $\varphi$ increases with respect to the parameter $\delta$. The effect of the increase in $\varphi$ dominates the effect of the decrease in GDA value, resulting in an overall increase in the certainty equivalent. Consequently, the GDA preference exhibits less risk aversion when $\delta$ becomes larger in $(1,\infty)$, as Theorem \ref{prop:C:delta}(ii) shows.

Similarly, the GDA value decreases with respect to the parameter $\beta$. In the case $\delta\in(0,1)$, the GDA value is equal to the certainty equivalent, indicating that the GDA preference exhibits more risk aversion as $\beta$ becomes larger.  In the case $\delta>1$, the function $\varphi$ increases with respect to the parameter $\beta$. The effect of the increase in $\varphi$ is dominated by the effect of the decrease in GDA value, resulting in an overall decrease in the certainty equivalent. Consequently, in this case, the GDA preference still exhibits more risk aversion when $\beta$ becomes larger, as Theorem \ref{prop:C:delta}(iii) shows.

\begin{remark}\label{rmk:rz}	In the case $\beta>0$ and $\delta>1$, for  the CRRA utility functions $U_{\rho}$, which are given by 
	\begin{align}\label{crra:utility}
		U_{\rho}(w)=\begin{cases}\frac{w^{1-\rho}}{1-\rho}, \quad\rho>0,\,\rho\neq1,\\
			\log w, \quad\rho=1,	\end{cases}
	\end{align}
	\citet[p. 1309]{Routledge2010} defines the time-$t$ certainty equivalent of $Y$ by the following equation:
	$$
	U_{\rho}(C_t(Y))=A\Big(\mE_t\left[U_{\rho}\left(Y\right)\right]
	-\beta\mE_t\left[\left(U_{\rho}(\delta C_t(Y))-U_{\rho}\left(Y\right)\right)_+\right]\Big),
	$$
	where $A=(1-\beta(\delta^{1-\rho}-1))^{-1}$ is the normalization that maintains the property that the certainty equivalent of a constant $w$ is $w$ itself. Their definition, however, does not apply to the logarithmic utility function ($\rho=1$) and non-CRRA utility functions.  Indeed, when $U(w)=\log w$, 
	then $C_t(w)=w$ implies that $A=(\log w)/(\log w-\beta \log \delta)$, which is absurd as $A$ depends on $w$.
	Moreover, monotonicity imposes another restriction that
	$A>0$, i.e., $\beta(\delta^{1-\rho}-1)<1$. Our definition applies to any utility function $U$ and is invariant under affine transformations of $U$: $C_t(Y)$ does not change if $U$ is replaced with 
	$\alpha_1 U+\alpha_0$, where $\alpha_1>0$ and $\alpha_0\in\mR$. The definition of \cite{Routledge2010}, however, does not have this invariance as $A\ne 1$. 
\end{remark}

We will frequently use the following regularity conditions on the utility functions. 

\begin{definition}\label{def:u:regular}
	For $n=0,1,2,\cdots$, we say that a utility function $U:(0,\infty)\to\mR$ is \emph{$n$-th-order regular} if $U$ is strictly increasing, $U\in C^n((0,\infty))$, and there exist constants $C>0$ and $\nu>0$ such that
	$$|U^{(0)}(w)|+\dots+|U^{(n)}(w)|\le C(w^\nu+w^{-\nu})\quad\forall\,w\in(0,\infty).
	$$
	Here $U^{(0)}=U$ and $U^{(n)}$ represents the $n$-th derivative of $U$, $n\ge1$. 
	Denote by $\cR_n$ the set of all $n$-th-order regular utility functions, $n\ge0$. 
\end{definition}

Obviously, CRRA utility functions are $n$-th-order regular for all $n\ge0$.

\subsection{Equilibrium strategy}\label{subsection:equli}
\noindent
For each $t\in[0,T]$, $p\in[1,\infty]$, and $m\ge 1$,  we use  $L^p(\F_t,\mR^m)$ to denote the set of all $L^p$-integrable, $\mR^m$-valued, and $\F_t$-measurable random variables. For simplicity, we write $L^{p}(\F_t)$ for $L^p(\F_t,\mR)$.  For $m\ge1$,  $L^0(\mF,\mR^m)$ is the space of $\mR^m$-valued, $\mF$-progressively measurable processes and $L^\infty(\mF,\mR^m)$ is the space of bounded processes in $L^0(\mF,\mR^m)$. 

A trading strategy is a process  $\pi=\{\pi_t, t\in[0,T)\}\in L^0(\mF,\mR^d)$ such that $\int_0^T\lVert\pi_t\rVert^2\md t<\infty$ a.s., where 
$\pi_t$ stands for the vector of portfolio weights according to which the wealth is invested into the stocks at time $t$.
The self-financing wealth process $\{W^{\pi}_t,0\leq t\leq T\}$ of a trading strategy $\pi$  satisfies the following stochastic differential equation (SDE):
\begin{equation}\label{wealthdynamic}
	\left\{
	\begin{aligned}
		&\md W^\pi_t=W^\pi_t \pi^{\top}_t\mu(t)\md t+W^\pi_t \pi^{\top}_t\sigma(t)\md B(t),\\	
		&W^\pi_0=w_0>0.
	\end{aligned}
	\right.
\end{equation}

Now we provide the definitions of admissible and equilibrium strategies.
\begin{definition}\label{adm-equi}
	A trading strategy $\pi$ is called \emph{admissible} if, for any $t\in [0,T)$,  $\mE_t\left[\left|U\left(W^\pi_T/W^\pi_t\right)\right|\right]<\infty$ a.s.  Denote by $\Pi$ the set of all admissible strategies. 
\end{definition}

Obviously, $L^\infty(\mF,\mR^d)\subset\Pi$. For $\pi\in\Pi$, the time-$t$ preference functional of the agent is given by $J(t,\pi):= \eta_t\left(W^\pi_T/W^\pi_t\right)$.  
Lemma \ref{lma:eta} implies that $J(t,\pi)$ is well defined for every $t\in[0,T)$ and $\pi\in\Pi$. 

\begin{remark}\label{rmk:equi}
	In the existing literature on GDA preferences, authors usually consider the  GDA value of the absolute wealth $W^\pi_T$, i.e., $J(t,\pi)=\eta_t\left(W^\pi_T\right)$.
	In contrast, in this paper we consider the GDA value of the relative wealth $W^\pi_T/W^\pi_t$, i.e., the gross return rate of the wealth from time $t$ to time $T$. 
	This formulation aligns with the intuition from behavioral economics that people concern the change of wealth level rather than the wealth level itself.
	It also makes the problem 
	tractable for general non-CRRA utility functions. In the case of CRRA utility functions, one can see that the two definitions are equivalent by the homogeneity of the CRRA utility functions. In the case $\beta=0$, as Remark \ref{rmk:gda} shows, the preference reduces to the classical EU preference, represented by $\mE_t\left[U\left(W^\pi_T/W^\pi_t\right)\right]$.  We 
	remark that even under this reduction, the problem is still time-inconsistent for non-CRRA utility because of the relative wealth.
\end{remark}

Hereafter, we always consider a fixed $\barpi\in\Pi$, which is a candidate equilibrium strategy. 
For any $t\in[0,T)$, $\epsilon\in(0,T-t)$ and $k\in L^\infty(\F_t,\mR^d)$,   let $\barpi^{t,\epsilon,k}:=\barpi+k\ind_{[t,t+\epsilon)}$, i.e.,  
\[\barpi^{t,\epsilon,k}_s=
\left\{
\begin{aligned}
	&\barpi_s+k, &s\in [t,t+\epsilon),\\
	&\barpi_s, &s\notin [t,t+\epsilon).
\end{aligned}
\right.
\]
$\barpi^{t,\epsilon,k}$ serves as a perturbation of $\barpi$.

Following \cite{Ekeland2010}, \cite*{hu2012time}, and \cite*{Bjork2017}, we introduce the definition of equilibrium strategies as follows.

\begin{definition}
	$\barpi$ is called an \emph{equilibrium strategy} if, for any $t\in[0,T)$ and $k\in  L^\infty(\F_t,\mR^d)$ such that $\barpi^{t,\epsilon,k}\in\Pi$ for all sufficiently small $\epsilon>0$, we have
	\begin{equation}\label{limsup:J}
		\lim_{\epsilon\downarrow 0 }\esssup_{\epsilon_0\in(0,\epsilon)}\frac{J(t,\bar{\pi}^{t,\epsilon_0,k})-J(t,\bar{\pi})}{\epsilon_0}\leq 0\quad \text{a.s.}
	\end{equation}
\end{definition}

We can also define the time-$t$ preference functional by $\hat{J}(t,\pi):= C_t\left(W^\pi_T/W^\pi_t\right)$. As $C_t=\varphi(\eta_t)$, where $\varphi$ is given by \eqref{eq:varphi}, we have $\hat{J}(t,\pi)=\varphi(J(t,\pi))$. If $U\in C^1((0,\infty))$ and $U'>0$, then $\varphi\in C^1((0,\infty))$ and $\varphi'>0$. In this case, replacing  $J$ with $\hat{J}$ in  \eqref{limsup:J} yields the same equilibrium strategies. 

\begin{remark}
	In literature, the equilibrium condition is usually 
	\begin{equation*}
		\limsup_{\epsilon\downarrow 0 }\frac{J(t,\bar{\pi}^{t,\epsilon,k})-J(t,\bar{\pi})}{\epsilon}\leq 0\quad \text{a.s.,}
	\end{equation*}
	the left hand side of which, however, may be non-measurable. 
	Therefore, we make a modification on account of the measurability.
\end{remark}

\section{Characterization of Equilibrium Strategies}\label{sec:char:equi}
\noindent
Let $\lambda(t) = (\sigma(t))^{-1}\mu(t)$ be the \emph{market price of risk}. Because all the market coefficients are deterministic, bounded, and right-continuous,  it is natural to conjecture that an equilibrium strategy $\barpi$ is in the form of
\begin{align}\label{solution}
	&\barpi_s=(\sigma^{\top} (s))^{-1}a_s,\quad s\in[0,T),
\end{align}
where $a$ is a deterministic, bounded, right-continuous $\mR^d$-valued function. Denote by $\mathcal{D}$ the trading strategies in the form of (\ref{solution}). It is obvious that $\mathcal{D}\subset L^\infty(\mF,\mR^d)\subset\Pi$.

For any given $t \in [0, T)$, let $v(t)$ and $y(t)$ respectively denote the \emph{cumulative risk} and the \emph{cumulative return} over time interval $[t,T)$ of the portfolio $\bar{\pi}$ in the form of (\ref{solution}). That is,
\begin{align*}
	v(t):=\int_t^T|\sigma^{\top}(s)\bar{\pi}_s|^2\md s=\int_t^T|a_s|^2\md s, \quad y(t):=\int_t^T\mu^{\top}(s)\bar{\pi}_s\md s=\int_t^Ta_s^{\top}\lambda(s)\md s,\quad t\in[0,T).
\end{align*}
It is easy to see that
$$  \text{conditionally on $\mathcal{F}_t$} ,\quad \frac{W^{\bar{\pi}}_T}{W^{\bar{\pi}}_t}\sim\logn\left(y(t)-{\frac 12}v(t),v(t)\right).$$ 

Let $g(v,y)$ be  the GDA value of a random variable $Z\sim\logn\left(y-v/2,v\right)$. By Lemma \ref{lma:eta}, $g(v,y)$ satisfies the following equation:
\begin{align*}
	U(g(v,y))=\mE\left[U(Z)\right]-\beta \mE\left[(U(\delta g(v,y))-U(Z))_+\right].
\end{align*}
The properties of the function $g:[0,\infty)\times\mR\to(0,\infty)$ are summarized in Lemmas \ref{lma:gCC^1} and \ref{lma:gCC^1:delta1}. The conclusions of these lemmas play a crucial role in characterizing the equilibrium strategies. Their proofs rely on the implicit function theorem and involve several technical derivations, which may appear cumbersome to some readers. For this reason, the detailed proofs are provided in Appendices \ref{sec:proof:gCC1} and \ref{sec:proof:gCC2}, and readers may safely skip them on a first reading without losing the main intuition of the paper.

\begin{lemma}\label{lma:gCC^1}
	Suppose that $\delta\neq1$, $U\in\mathcal{R}_2$, and $U''<0$. Then $g\in C^1([0,\infty)\times \mR)$,
	$g_v<0$  and $g_y>0$ on $[0,\infty)\times \mR$, and for all $(v,y)\in(0,\infty)\times \mR$, 
	\begin{align}
		\frac{g_y(v,y)}{ g_v(v,y)}=\frac{2\mE  \left[U'(Z)Z\left(1+\beta\ind_{\{Z<\delta g(v,y)\}}\right)\right]}{{\mE \left[U''(Z)Z^2\left(1+\beta\ind_{\{Z<\delta g(v,y)\}}\right)\right]-\beta U'(\delta g(v,y))\delta g(v,y) N'\left(\frac {\log(\delta g(v,y))-y+{v\over2}}{\sqrt{v}}\right)/{\sqrt{v}}}},\nonumber\end{align} 
	where $Z\sim\logn\left(y-v/2,v\right)$ and $N$ is the standard normal distribution function.
\end{lemma}
\begin{lemma}\label{lma:gCC^1:delta1}
	Suppose that $\beta>0$, $\delta=1$ and $U\in\mathcal{R}_1$ with $U'(1)>0$. Then $\lim\limits_{v\downarrow0,y/\sqrt{v}\to0}(g(v,y)-g(0,0))/\sqrt{v}=c^*$, where $c^*<0$ is the unique solution of the  equation: $			c+\beta cN(c)+\beta N'(c)=0.$ Furthermore, if $U$ is also concave, then $g\in  C^1((0,\infty)\times \mR)$, $g_v<0$  and $g_y>0$ in $(0,\infty)\times \mR$,  and $\lim\limits_{v\downarrow0,y/\sqrt{v}\to0}g_y(v,y)/\bigl(\sqrt{v}g_v(v,y)\bigr)=2/c^*$.
\end{lemma}

Obviously, $J(t,\bar{\pi})=g(v(t),y(t))$.  

We are now going to derive a sufficient and necessary condition for \eqref{limsup:J}. It is easy to see that $\barpi^{t,\epsilon,k}\in L^\infty(\mF,\mR^d)\subset\Pi$ for all $k\in  L^\infty(\F_t,\mR^d)$ and $\epsilon\in(0,T-t)$.
For a perturbation $\barpi^{t,\epsilon,k}$  of  $\bar{\pi}$,  the cumulative risk and the cumulative return  are  
\begin{align*}
	\tilde{v}(t)=v(t)+\int_t^{t+\epsilon}|\sigma^{\top}(s)k|^2\md s+2\int_t^{t+\epsilon}k^{\top}\sigma(s)\sigma^{\top}(s)\bar{\pi}_s\md s, \quad \tilde{y}(t)=y(t)+\int_t^{t+\epsilon}k^{\top}\mu(s)\md s.
\end{align*}
As $k$ is $\mathcal{F}_t$-measurable,  by \eqref{eq:etat(Y)}, we have $J(t,\barpi^{t,\epsilon,k})=g(\tilde{v}(t),\tilde{y}(t))$.
Thus, Condition \eqref{limsup:J} is equivalent to
\begin{align}
	0\ge&\lim_{\epsilon\to0}\frac{g(\tilde{v}(t),\tilde{y}(t))\!-\!g(v(t),y(t))}{\epsilon}\nonumber\\
	=&g_v(v(t),y(t))\Big(|\sigma^{\top}(t)k|^2\!+\!2k^{\top}\sigma(t)\sigma^{\top}(t)\bar{\pi}_t\Big)\!+\!g_{y}(v(t),y(t))k^{\top}\sigma(t)\lambda(t),\label{eq:limsup}
\end{align}
provided that $g$ is differentiable at $(v(t),y(t))$.

\subsection{The case \texorpdfstring{$\delta\neq 1 $}{Lg}}\label{sec:MRS}
\noindent
We first consider the case $\delta\ne1$.  In this subsection, we always assume $U\in\mathcal{R}_2$ and $U''<0$ unless otherwise stated. According to Lemma \ref{lma:gCC^1}, $g$ is $C^1$ on $[0,\infty)\times\mR$.
Then we know that $\barpi$ is an equilibrium strategy if and only if \eqref{eq:limsup} holds for all $t\in[0,T)$ and $k\in L^\infty(\F_t,\mR^d)$. Observe that the right hand side of  \eqref{eq:limsup} is quadratic in $\sigma^{\top}(t)k$ and $\sigma(t)$ is invertible.
Therefore, for every $t\in[0,T)$, \eqref{eq:limsup} holds for all $k\in L^\infty(\F_t,\mR^d)$ if and only if
\begin{align}
	\begin{cases}
		g_v(v(t),y(t))\leq0,\\
		2\sigma^{\top}(t)\bar{\pi}_tg_v(v(t),y(t))+\lambda(t)g_y(v(t),y(t))=0.
	\end{cases}\label{tradeoff}
\end{align}
By Lemma \ref{lma:gCC^1}, $g_v<0$. Then by (\ref{solution}), \eqref{tradeoff} is equivalent to
\begin{align}\label{MRS_equli}
	a_t=-\frac{g_y(v(t),y(t))}{2g_v(v(t),y(t))}\lambda(t)=\frac{1}{2\textup{MRS}_{v,y}(v(t),y(t))}\lambda(t),
\end{align}
where $\textup{MRS}_{v,y}(v(t),y(t)):= -g_v(v(t),y(t))/g_y(v(t),y(t))$ is the \emph{marginal rate of substitution} (MRS) at $(v(t),y(t))$ of the cumulative risk for the cumulative return. 

The above discussion can be concluded by the following theorem.  
\begin{theorem}
	\label{thm:deltaneq} Suppose that $\delta\neq 1$.
	Let $\bar{\pi}\in\mathcal{D}$ be given by (\ref{solution}).  Then $\bar{\pi}$
	is an equilibrium if and only if
	\begin{align}\label{eq:ode2}
		a_t=-\frac{g_y(v(t),y(t))}{2g_v(v(t),y(t))}\lambda(t), \quad t\in[0,T).
	\end{align}
\end{theorem}
\begin{remark}
Theorem \ref{thm:deltaneq} shows that the equilibrium strategy $\bar{\pi}$ is determined by the MRS between cumulative risk $v(t)$ and cumulative return $y(t)$. Intuitively, at each time $t$, the investor adjusts her/his strategy so that the trade-off between risk and return aligns with her/his GDA preferences. The equilibrium investment $(\sigma^{\top}(t))^{-1}a_t$ is proportional to $(\sigma^{\top}(t))^{-1}\lambda(t)$, with the proportionality factor given by the MRS, reflecting how additional expected return compensates for extra risk.
\end{remark}

When $\beta=0$ and the utility function $U$ is the CRRA utility function $U_{\rho}$, it is evident that $2\textup{MRS}_{v,y}=\rho$, where $\rho$ is the coefficient of relative risk aversion of $U_{\rho}$. In this case, the equilibrium strategy coincides with the Merton solution $\pi^*(t)=(1/\rho)(\sigma^{\top})^{-1}(t)\lambda(t)$. In general, MRS tells us the additional amount of the cumulative return that the investor must be given  as a compensation for a one-unit marginal addition in the cumulative risk while maintaining on the same indifference curve of $g$. In particular, Lemma \ref{lma:gCC^1} implies that       
$2\textup{MRS}_{v,y}(0,0)=	-U''(1)/U'(1)$. Therefore, $2\textup{MRS}_{v,y}(v,y)$ can be regarded as the coefficient of risk aversion of $g$ at $(v,y)$. 
Consequently,  in the general case, the equilibrium strategy as determined by \eqref{eq:ode2} is a Merton-like solution with the coefficient of relative risk aversion replaced with the coefficient of risk aversion of $g$.

\begin{remark}\label{remark:converges}
	If $\bar{\pi}=(\sigma^{\top})^{-1}a\in\mathcal{D}$ is an equilibrium,  by Lemma \ref{lma:gCC^1}, we have 
	\begin{align*}
		\lim_{t\to T}-\frac{g_y(v(t),y(t))}{2g_v(v(t),y(t))}=	-	\frac{U'(1)}{U''(1)}.
	\end{align*}
	Furthermore, if $U$ is the CRRA utility function $U_{\rho}$, then $-U'(1)/U''(1)=1/\rho$. In this case, the equilibrium investment converges to the Merton solution as the terminal time is approaching.
\end{remark}

\begin{figure}[ht]
	\centering
	\begin{subfigure}{0.48\textwidth}
		\centering
		\includegraphics[width=\linewidth]{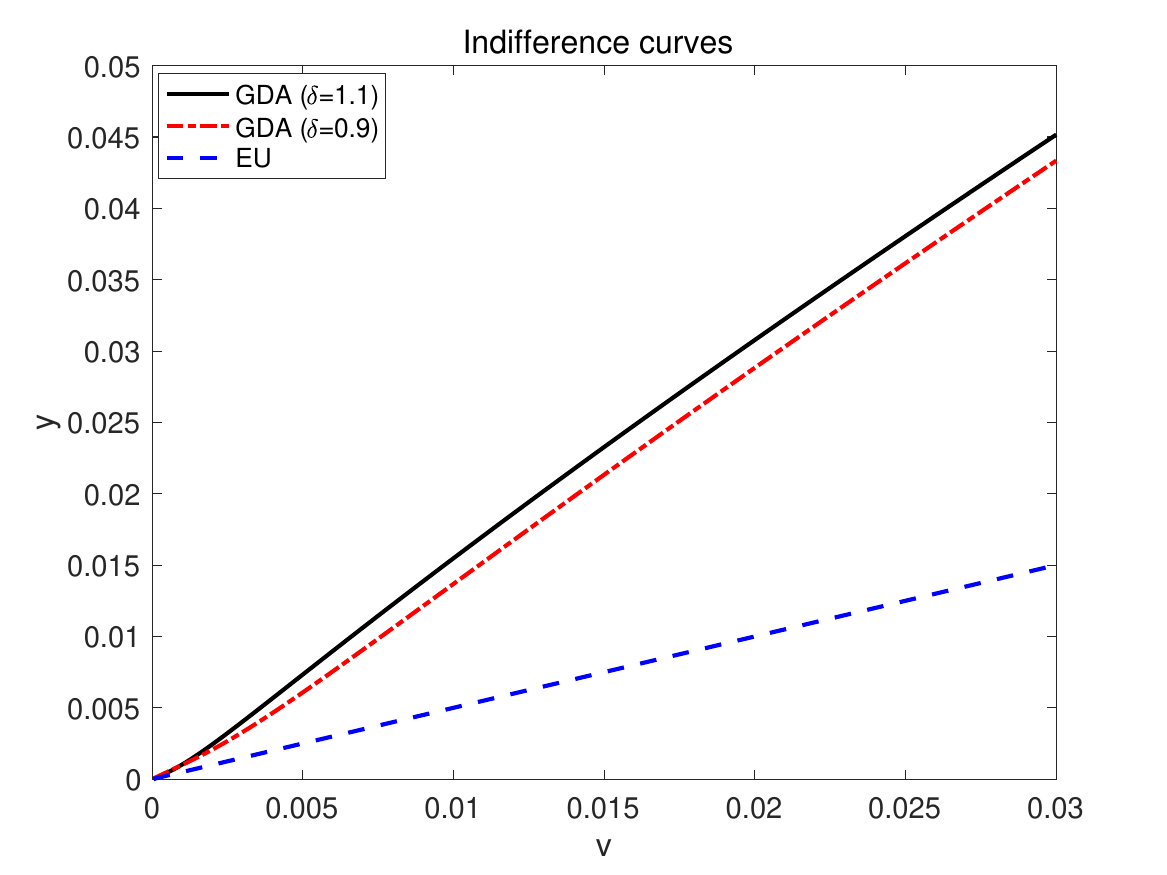}
		\caption{Indifference curves}
	\end{subfigure}
	\begin{subfigure}{0.48\textwidth}
		\centering
		\includegraphics[width=\linewidth]{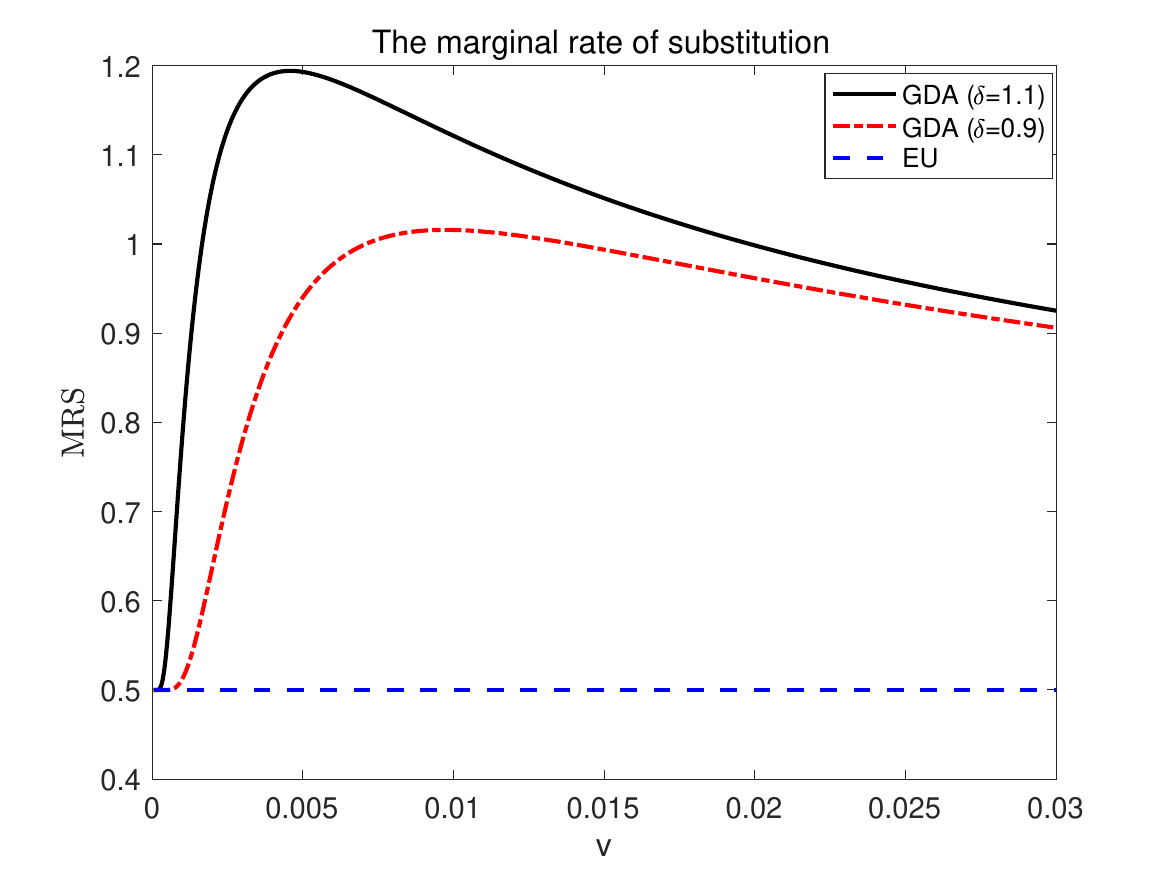}
		\caption{Marginal rate of substitution}
	\end{subfigure}
	\caption{The solid, dash-dot, and dashed lines represent the GDA ($\beta=0.5$, $\delta=1.1$),  GDA ($\beta=0.5$, $\delta=0.9$) and EU  ($\beta=0$) preferences, respectively. The utility function $U$ is the CRRA utility function $U_{\rho}$ with $\rho=1$.
	}\label{ind_curve}
\end{figure}

As an illustration, 
Fig. \ref{ind_curve}  displays the indifference curves (panel (a)) and the MRS (panel (b)) for GDA and EU preferences with the CRRA utility function $U=U_\rho$. By the homogeneity of the CRRA utility  functions, we know that $g(v,y)=e^yg(v,0)$ for all $(v,y)\in[0,\infty)\times\mR$. Then, fixing the preference, i.e., fixing the parameters $\rho$, $\beta$, and $\delta$, the indifference curves $(g=c)$, $c>0$, are parallel and $\textup{MRS}_{v,y}(v,y)=-g_v(v,0)/g(v,0)$ does not depend on $y$.  So we only plot the indifference curve that passes through the origin for each preference. 

The indifference curves of GDA and EU preferences are increasing in the cumulative risk. This is natural because the preferences are risk averse. Furthermore, the indifference curve of the GDA preference remains above that of EU preference. This is due to the fact that the GDA-agent is disappointment averse, indicating that,  in order to be indifferent to the origin $(0,0)$, with the same level of cumulative risk, the GDA-agent requires a higher return than the EU-agent. It should be noted that the indifference curves of the GDA preferences are tangential to that of the EU preference at the origin $(v,y)=(0,0)$,
implying that the GDA preferences generate the asymptotically same risk aversion of the EU preference when the risk is very small.

Under EU, the MRS is constant and is represented by a horizontal line.  Under GDA, the MRS exhibits an inverse-U shape:  when the cumulative risk is small, the MRS of the GDA preference increases from the constant MRS of the EU preference;
when the cumulative risk is large, the MRS of the GDA preference decreases.

From Theorem \ref{thm:deltaneq}, the equilibrium strategies are characterized in terms of the solutions to Equation \eqref{eq:ode2}. We are now going to investigate Equation \eqref{eq:ode2}. For simplicity of notation, we define the function $m:[0,\infty)\times \mR\to(0,\infty)$ by 
\begin{align}
	&m(x,y)=-\frac{g_y(x^2,y)}{2g_v(x^2,y)}\nonumber\\
	&=\frac{\mE  \left[U'(Z)Z\left(1+\beta\ind_{\{Z<\delta g(x^2,y)\}}\right)\right]}{{-\mE \left[U''(Z)Z^2\left(1+\beta\ind_{\{Z<\delta g(x^2,y)\}}\right)\right]+\beta U'(\delta g(x^2,y))\delta g(x^2,y) N'\left(\frac {\log(\delta g(x^2,y))-y+{x^2\over2}}{x}\right)/{x}}},
	\label{eq:new:m}
\end{align}
where $Z\sim\logn\left(y-x^2/2,x^2\right)$. Then Equation \eqref{eq:ode2} is equivalent to the following integral equation:
\begin{align}\label{ode:m}
	a_t=m\left(\sqrt{\int_t^{T}|a_s|^2\md s},\int_t^Ta_s^{\top}\lambda(s)\mathrm{d}s\right)\lambda(t), \quad t\in[0,T).    \end{align}
The following theorem establishes the existence and uniqueness of the solution to the integral equation \eqref{ode:m}.
\begin{theorem}\label{ode2:solution}
	Suppose  that  $m$ is bounded and locally Lipschitz continuous on $[0,\infty)^2$.
	Then Equation \eqref{ode:m}, and consequently Equation $(\ref{eq:ode2})$,  have a unique solution in $L^{\infty}(0,T)$.
\end{theorem}
\textbf{Proof.} 	From the boundedness of $m$, $m\leq C_0$ for some $C_0>0$. Then, by $(\ref{ode:m})$, any solution $a$ is  bounded by $C_0\|\lambda\|_{\infty}$. Define $\mathcal{S}=\{a\in L^{\infty}(0,T):  \|a\|_{\infty} \leq C_0\|\lambda\|_{\infty}\}$, which is a closed subset of $L^{\infty}(0,T)$. We can choose  $M>0$  such that $\left(\sqrt{\int_t^{T}|a_s|^2\md s},\int_t^{T}a_s^{\top}\lambda(s)\mathrm{d}s\right)\in[0,M]^2$ for all $a\in\mathcal{S}$ and $t\in[0,T)$. Let $L>0$ be the Lipschitz constant of $m$ on $[0,M]^2$. 
	
	We first solve $(\ref{ode:m})$ on the interval $[T-\epsilon,T)$, where $0<\epsilon\leq \min\{1,T\}$ is to be determined. Consider the operator: $\Ti: \mathcal{S}(T-\epsilon,T)\to \mathcal{S}(T-\epsilon,T)$ given by$$(\Ti a)_t:= \lambda(t)m\left(\sqrt{\int_t^{T}|a_s|^2\md s},\int_t^{T}a_s^{\top}\lambda(s)\mathrm{d}s\right),\quad t\in[T-\epsilon,T).$$ For any $a^{(i)}\in \mathcal{S}(T-\epsilon,T)$, $i=1,2$, we have 
	\begin{align*}
		&\|\Ti a^{(1)}-\Ti a^{(2)} \|_{\infty}\\
		&\leq L \|\lambda\|_{\infty}\sup_{t\in[T-\epsilon,T]}\left[\left|\sqrt{\int_t^{T}|a^{(1)}_s|^2\md s}\!-\!\sqrt{\int_t^{T}|a^{(2)}_s|^2\md s}\right|\!+\!\left|\int_t^{T}(a^{(1)}_s)^\top\lambda(s)\md s\!-\!\int_t^{T}(a^{(2)}_s)^\top\lambda(s)\md s\right|\right]\\
		&\leq L\|\lambda\|_{\infty} \sup_{t\in[T-\epsilon,T]}\left[\sqrt{\int_t^{T} |a^{(1)}_s-a^{(2)}_s|^2\md s}+\|\lambda\|_{\infty}\int_t^{T} |a^{(1)}_s-a^{(2)}_s|\md s\right] \\
		&\leq 2L\|\lambda\|_{\infty}\sqrt{\epsilon}\|a^{(1)}-a^{(2)}\|_{\infty}(1+\|\lambda\|_{\infty}).
	\end{align*}
	Choosing $\epsilon<\min \left\{1,T,1/\bigl(4L^2\|\lambda\|_{\infty}^2(1+\|\lambda\|_{\infty})^2\bigr) \right\}$, we see that $\Ti$ is a contraction on $L^{\infty}(T-\epsilon,T)$. Thus there exists a unique fixed point of $\Ti$, which is the solution to (\ref{ode:m}) on $[T-\epsilon,T)$. Finally, we consider the partition $0=t_0<t_1<\cdots<t_N=T$ such that $|t_k-t_{k-1}|\leq \epsilon$ for each $k=1,2,\cdots,N$. Then we have found a unique $a$ such that (\ref{ode:m}) is satisfied on $[t_{N-1},t_N)$. Suppose that we have constructed such $a$ on $[t_k,t_N)$. Consider the operator $\Ti_k:\mathcal{S}(t_{k-1},t_k)\to \mathcal{S}(t_{k-1},t_k)$ given by $$(\Ti_k b)_t:= \lambda(t)m\left(\sqrt{\int_t^{T} \left|\zeta_{b}(s)\right|^2\md s },\int_t^{T} \zeta_{b}(s)^{\top}\lambda(s)\md s \right),\quad t\in[t_{k-1},t_k)$$ where
	\begin{align*}
		\zeta_{b}(s)=\begin{cases}
			b_s, &s\in [t_{k-1},t_k),\\
			a_s, &s\in [t_k,t_N).
		\end{cases}
	\end{align*}	
	Then, using the same contraction argument, we can find a unique fixed point $b^*$ of $\Ti_k$, and $\zeta_{b^*}$ uniquely solves (\ref{ode:m}) on $[t_{k-1},t_N)$. By induction, (\ref{ode:m}) has a unique solution on $[0,T)$.\hfill $\square$ 

\begin{remark}\label{rmk:ode:secondway}
	There is another possible way to deal with Equation \eqref{ode:m} from the viewpoint of ordinary differential equation (ODE). Noting that  Equation \eqref{ode:m} is equivalent to the following two-dimensional ODE:
	\begin{align*}
		\begin{cases}
			v'(t)=-m^2(\sqrt{v(t)},y(t))|\lambda(t)|^2, \quad&v(T)=0,\\
			y'(t)=-m(\sqrt{v(t)},y(t))|\lambda(t)|^2,\quad &y(T)=0.
		\end{cases}
	\end{align*}
	Let $G(x,y)=m(\sqrt{x},y)$. To apply the standard theory to the above ODE, it requires the local Lipschitz continuity of $G$, which holds only if $$\limsup_{x\downarrow0}|G_x(x,y)|=\limsup_{x\downarrow0}\left|\frac{m_x(\sqrt{x},y)}{ 2\sqrt{x}}\right|<\infty\quad\forall\,y\in\mathbb{R}.$$	Therefore, it requires more regularity conditions on $m_x$ and, consequently, on $U$,  to apply the standard theory to the above ODE. 
\end{remark}

Next, let us examine the conditions of Theorem \ref{ode2:solution}. First, $m$ is locally Lipschitz continuous if $m$ is continuously differentiable. Lemma \ref{lma:mC1} below shows that $m\in C^1([0,\infty)\times\mR)$ if  we further assume $U\in\mathcal{R}_3$.

\begin{lemma}\label{lma:mC1}
	Suppose further that  $U\in\mathcal{R}_3$. Then $m\in C^1([0,\infty)\times\mR)$.
\end{lemma}
\begin{proof}
    See Appendix \ref{sec:proof:mC1}.
\end{proof}

The following assumption implies the boundedness of $m$.
\begin{assumption}\label{assump:m:2}
	There is a constant $C_0>0$ such that $-xU''(x)/U'(x)\geq 1/C_0$ for all $x>0$.
\end{assumption}

If Assumption \ref{assump:m:2} holds, it is easy to see that $m$ is bounded by $C_0$ from \eqref{eq:new:m}.
Furthermore, suppose that  $U\in\mathcal{R}_3$,  then $m$ is locally Lipschitz continuous by Lemma \ref{lma:mC1}. From Theorem \ref{ode2:solution}, Equation  $(\ref{ode:m})$  has a unique solution $a\in L^{\infty}(0,T)$. Then, based on Theorem \ref{thm:deltaneq}, we have the following theorem.
\begin{theorem}\label{thm:conclusions}
	Suppose that  $U\in\mathcal{R}_3$, $U''<0$, and  Assumption  \ref{assump:m:2} holds. Then Equation  $(\ref{ode:m})$  has a unique solution $a\in L^{\infty}(0,T)$. Moreover
	\begin{align*}
		\bar{\pi}_s:=(\sigma^{\top}(s))^{-1}a_s,\quad &s\in[0,T),
	\end{align*}
	is the unique equilibrium in $\mathcal{D}$.
\end{theorem}

\subsection{The case \texorpdfstring{$\delta=1 $}{Lg}}
\noindent
In this subsection, we explore the equilibrium strategies under DA preferences (i.e., $\delta=1$). We are going to show that $\barpi=0$ stands as the unique equilibrium in $\mathcal{D}$, signifying non-participation in the stock market. 

\begin{theorem}\label{barpi=0} 
	Suppose that $\beta>0$, $U\in \mathcal{R}_1$ with $U'(1)>0$. Then
	$\bar{\pi}=0$ is an equilibrium. Moreover, if $U$ is concave, then $\barpi=0$ is the unique equilibrium  in $\mathcal{D}$.
\end{theorem}
\textbf{Proof.} 	For $\bar{\pi}=0$ , given $t\in[0,T)$, we have
	\begin{align*}
		&J(t,\bar{\pi})=g(0,0)=1,\\
		&J(t,\bar{\pi}^{t,\epsilon,k})=g\left(\int_t^{t+\epsilon}|\sigma^{\top}(s)k|^2\md s,\int_t^{t+\epsilon}k^{\top}\sigma(s)\lambda(s)\mathrm{d}s\right).
	\end{align*}
	Using  Lemma \ref{lma:gCC^1:delta1}, we obtain
	\begin{align*}
		\lim_{\epsilon\downarrow0}\frac{g\left(\int_t^{t+\epsilon}|\sigma^{\top}(s)k|^2\md s,\int_t^{t+\epsilon}k^{\top}\sigma(s)\lambda(s)\mathrm{d}s\right)-g(0,0)}{\sqrt{\epsilon}}=|\sigma^{\top}(t)k|c^*.
	\end{align*}
	Therefore, 
	\begin{align*}
		\lim_{\epsilon\downarrow0}\esssup_{\epsilon_0\in(0,\epsilon)}\frac{J(t,\barpi^{t,\epsilon_0,k})-J(t,\barpi)}{\epsilon_0}\leq 0.
	\end{align*}
	Thus, $\barpi=0$ is an equilibrium.
	
	Suppose further that $U$ is concave, by Lemma \ref{lma:gCC^1:delta1}, $g$ is $C^1$ in $(0,\infty)\times\mR$.
	Let $\bar{\pi}\in\mathcal{D}$ be given by (\ref{solution}).
	In this case, $g$ is differentiable at $(v(t),y(t))$ for all $t\in[0,T_0)$, where
	$$T_0=T_0(\bar{\pi}):= \inf\{t\in[0,T]: v(t)=0\}.$$
	Assume that $\bar{\pi}$ is an equilibrium. 
	An analogue of Theorem \ref{thm:deltaneq} yields 
	\begin{align*}
		a_t=-\frac{g_y(v(t),y(t))}{2g_v(v(t),y(t))}\lambda(t), \quad t\in[0,T_0).
	\end{align*}
	Then $v$ satisfies the following ODE:
	\begin{align*}
		v'(t)=-\left|\frac{g_y(v(t),y(t))}{2\sqrt{v(t)}g_v(v(t),y(t))}\right|^2|\lambda(t)|^2v(t),\quad t\in[0,T_0),\quad v(T_0)=0.
	\end{align*}
	By Lemma \ref{lma:gCC^1:delta1}, $\left|g_y(v(t),y(t))/\bigl(2\sqrt{v(t)}g_v(v(t),y(t))\bigr)\right|^2$ is bounded. Thus, Gronwall's lemma gives that $v\equiv0$, and consequently, $a\equiv0$.\hfill $\square$

\begin{figure}[ht]
	\centering
	\begin{subfigure}{0.48\textwidth}
		\centering
		\includegraphics[width=\linewidth]{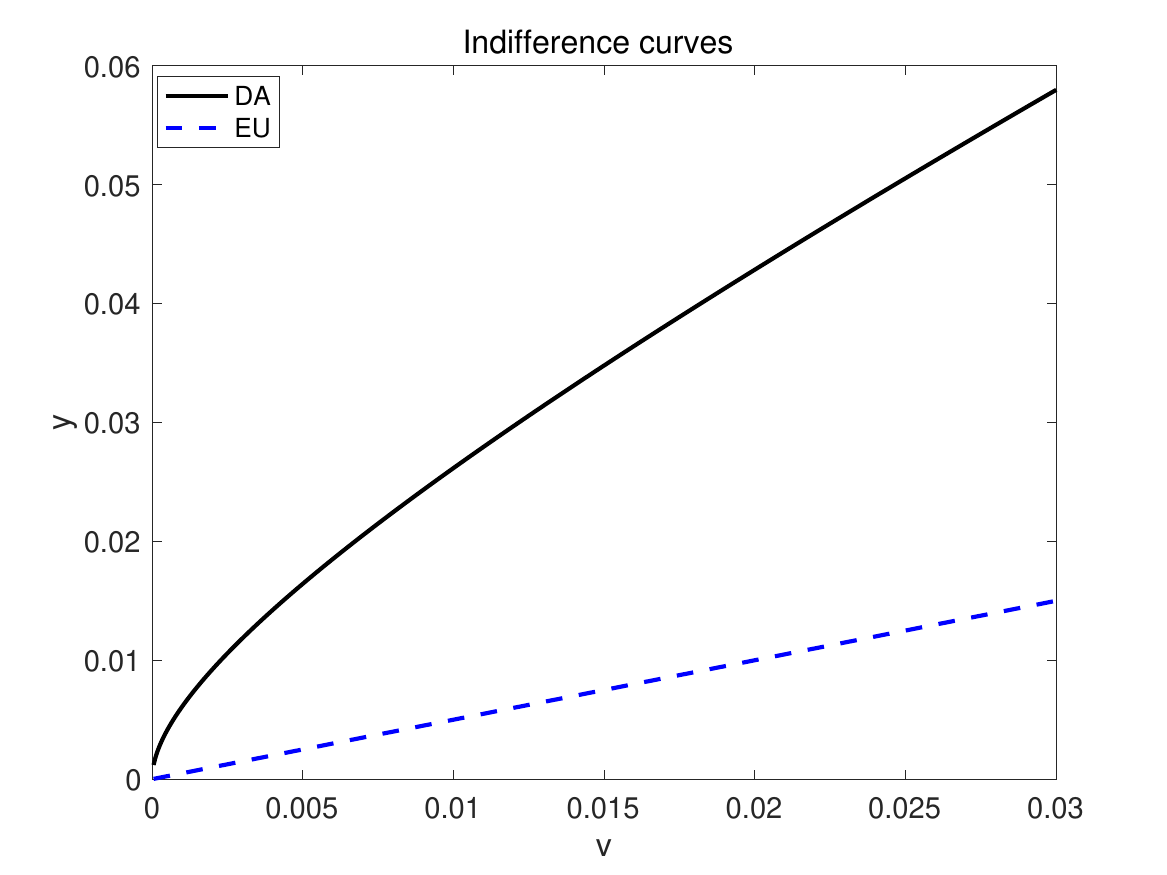}
		\caption{Indifference curves}
	\end{subfigure}
	\begin{subfigure}{0.48\textwidth}
		\centering
		\includegraphics[width=\linewidth]{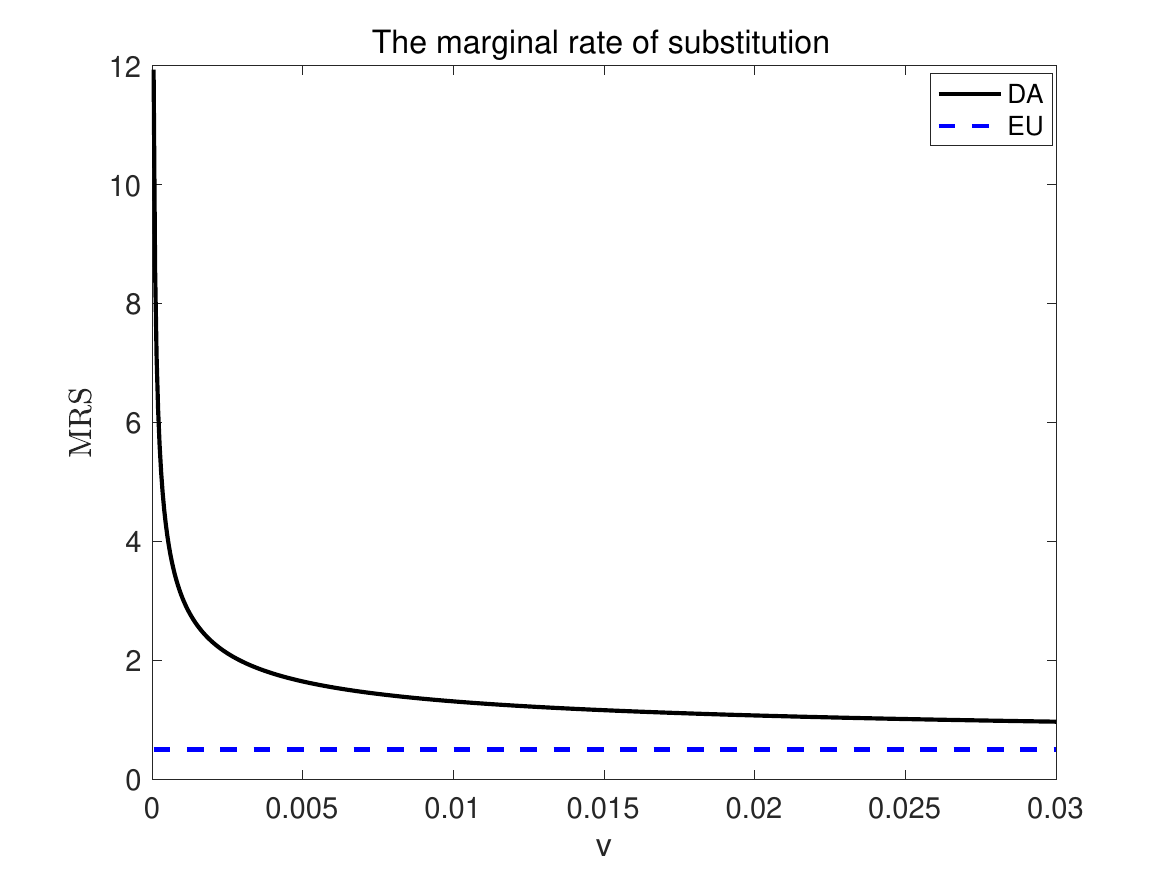}
		\caption{Marginal rate of substitution}
	\end{subfigure}
	\caption{The solid and  dashed lines represent the DA ($\beta=0.5$, $\delta=1$)  and EU  ($\beta=0$) preferences, respectively. The utility function $U$ is the CRRA utility function $U_{\rho}$ with $\rho=1$.}\label{ind_curve_delta_1}
\end{figure}

By Lemma \ref{lma:gCC^1:delta1}, when $\delta=1$, we have $\lim\limits_{v\downarrow0,y/\sqrt{v}\to0}g_y(v,y)/\bigl(\sqrt{v}g_v(v,y)\bigr)=2/c^*$, which implies that the MRS of the DA preference is approximately equal to  $-c^*/(2\sqrt{v})$ when $v$ is small, confirming a numerical observation of \cite*[Example 5]{Backus2004}.  Such behavior is referred to as \emph{first-order risk aversion}, as the risk premium of a small gamble is approximately $\int_0^v-c^*/(2\sqrt{s})\md s=-c^*\sqrt{v}$, which is proportional to the small gamble's standard deviation $\sqrt{v}$. (An EU preference usually exhibits second-order risk aversion.) As a consequence of first-order risk aversion,  non-participation in the stock market is the unique equilibrium  in $\mathcal{D}$; cf. the proof of Theorem \ref{barpi=0}.  

In contrast, when $\delta\ne1$, as Theorem \ref{thm:deltaneq} shows, for an equilibrium $\barpi\in\mathcal{D}$, the equilibrium investment $\bar{\pi}_t$ is zero only if the market price $\lambda(t)$ is zero, that is, the expected return $\mu(t)$ is zero.
Conversely,  as long as the expected return on stocks is non-zero, the agent will invest a non-zero amount in stocks, as the MRS is always finite and positive. This aligns with an observation of \cite*{dahlquist2017asymmetries} on a discrete-time model. 

Fig. \ref{ind_curve_delta_1} displays the indifference curves  that pass through the origin (panel (a)) and the MRS (panel (b)) for DA and EU preferences with the CRRA utility function $U=U_\rho$. It should be noted that the indifference curve of the DA preference is tangential to the $y$-axis instead of the indifference curve of the EU preference at $(v,y)=(0,0)$.  The MRS of the DA preference is infinite at $(v,y)=(0,0)$,  indicating that the DA preference generates very large risk aversion (first-order risk aversion) when the risk is small.  
In contrast to the MRS of a GDA preference, which is inverse-U shaped, the MRS of the DA preference is decreasing.

\section{Special case: CRRA utility}\label{sec:CRRA}
\noindent
In this section,  we consider the special case when $U$ is the CRRA utility function $U_\rho$ given by \eqref{crra:utility}. We focus on the case $\delta\neq 1$ because $\barpi=0$ is the unique equilibrium  in $\mathcal{D}$ in the case $\beta>0$ and $\delta=1$.

Using the homogeneity of $U_\rho$, we know that  $g(v,y)=e^yg(v,0)$ for all $(v,y)\in[0,\infty)\times\mR$. Then  $\textup{MRS}_{v,y}(v,y)=-g_v(v,0)/g(v,0)$ does not depend on $y$. 
Neither does $m(x,y)$:
$$m(x,y)=m(x,y') \quad\forall\,y,y'\in\mathbb{R}.$$
Abusing notation, let $m(x)=m\left(x,x^2/2\right)$, and $g(v)=g\left(v,v/2\right)$, $x\in[0,\infty) ,v\in[0,\infty)$. 
Then, from \eqref{eq:new:m}, it is easy to verify that
\begin{align*}
	m(x)=\frac{x}{\rho x+\frac{\beta N'\left(\frac {\log (\delta g(x^2))}{x}-(1-\rho)x\right)}{1+\beta N\left( \frac{\log (\delta g(x^2))}{x}-(1-\rho)x\right)}}=\frac{x}{\rho x+G'\left(\frac{\log (\delta g(x^2))}{x}-(1-\rho) x\right)},
\end{align*}
where $G(z)=\log (1+\beta N(z))$, and $g$ satisfies
\begin{align}\label{g:crra}
	\!\!\!\!\!\!\!\!\!\!\!\!\!\!\!\!\!\!\!\!\!	\begin{cases}
		\log (\delta g(x^2))\!\!-\!\!\log \delta+\beta N\left(\frac{\log (\delta g(x^2))}{x}\right)\log (\delta g(x^2))\!\!+\!\!\beta x N'\left(\frac{\log (\delta g(x^2))}{x}\right)=0,  &\rho=1,\\
		(\delta g(x^2))^{1-\rho}\left(\delta^{\rho-1}\!+\!\beta N\left(\frac{\log (\delta g(x^2))}{x}\right)\right)\!=\!\mathrm{e}^{\frac 12 (1-\rho)^2x^2}\left(1+\beta N\left( \frac{\log (\delta g(x^2))}{x}\!-\!(1\!\!-\!\!\rho)x\right)\right), &\rho\neq1.
	\end{cases}
\end{align}
Obviously, $U_{\rho}$ satisfies the conditions of Theorem \ref{thm:conclusions}. Then the equilibrium in $\mathcal{D}$ exists uniquely. Moreover, let the equilibrium have the form (\ref{solution}). Then $a$ satisfies the following equation:
$$a_t=m\left(\sqrt{\int_t^{T}|a_s|^2\md s} \right)\lambda(t), \quad t\in[0,T),$$
which implies  that $v$ satisfies the following ODE
$$v'(t)=-m^2\left(\sqrt{v(t)}\right)|\lambda(t)|^2,\quad t\in[0,T),\quad v(T)=0.$$
It is an ODE with separated variables and its solution is given by
$$v(t)=\mathcal{M}^{-1}\left(\int_t^{T}|\lambda(s)|^2\md s\right),\quad t\in[0,T],$$
where $\mathcal{M}^{-1}$ is the inverse function of 
$$\mathcal{M}(x)=\int_0^x\frac{1}{m(\sqrt{y})^2}\md y,\quad x\ge0.$$
Therefore, 
\begin{align}\label{closed-form-a}
	a_t=
	m\left(\sqrt{\mathcal{M}^{-1}\left(\int_t^{T}|\lambda(s)|^2\md s\right)}\right) \lambda(t),\quad t\in[0,T).
\end{align}
The above discussion can be concluded by 
the following theorem.
\begin{theorem}
	For the CRRA utility function $U=U_\rho$, let  $a$ be given by  \eqref{closed-form-a} and $\bar{\pi}=(\sigma^{\top})^{-1}a$.
	Then $\barpi$ is the unique equilibrium in $\mathcal{D}$.
\end{theorem}
\begin{remark}\label{rmk:limitpart}
	Under the EU preference ($\beta=0$), we have $m(x)\equiv 1/\rho$ and hence $\bar{\pi}_t=(1/\rho)\allowbreak{(\sigma^{\top}(t))^{-1}}\lambda(t)$, which precisely coincides with the optimal investment proportion derived in \cite{Merton1971}. When $\beta>0$, it is evident that $0<m(x)< 1/\rho$ for all $x>0$ and hence the equilibrium portfolio is more conservative than the Merton solution. It is natural because the GDA preference exhibits more risk aversion than the EU preference, as shown by the discussion on the MRS of $g$ in Subsection \ref{sec:MRS}. 
\end{remark}

\subsection{Numerical analysis}\label{subsec:Nume-analy}
\noindent
In this subsection, we conduct some numerical analysis to  study the effects of $\delta$ and $\beta$ on the equilibrium trading strategies.

We consider a simple Black-Scholes market model with one risky asset,  whose volatility is $\sigma\equiv 0.3$ and expected return rate $\mu\equiv0.06$. The utility function $U$ is the CRRA utility function $U_\rho$ with the relative risk aversion coefficient $\rho=1$, that is, $U(w)=\log w$. Finally, the time horizon is $T=3$.

\begin{figure}[ht]
	\centering
	\begin{subfigure}{0.48\textwidth}
		\centering
		\includegraphics[width=\linewidth]{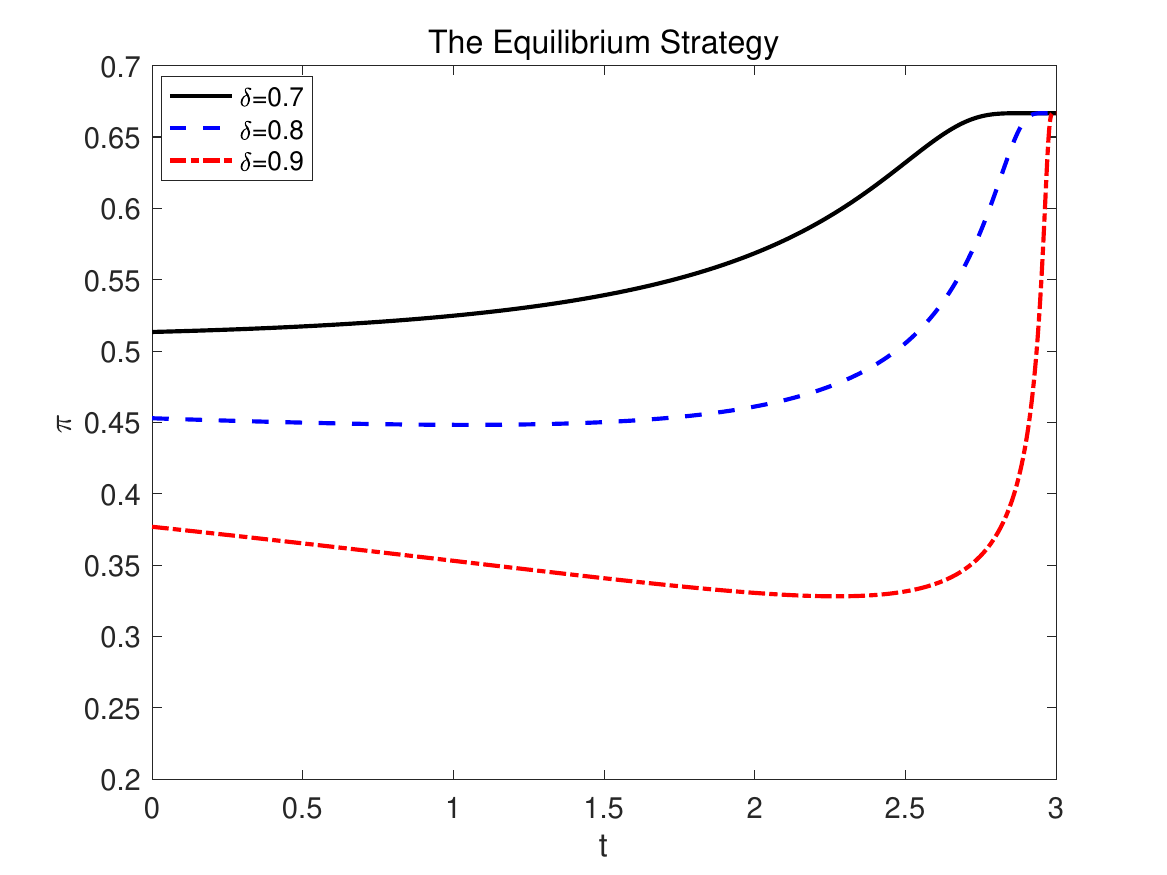}
		\caption{$\delta<1$}
	\end{subfigure}
	\begin{subfigure}{0.48\textwidth}
		\centering
		\includegraphics[width=\linewidth]{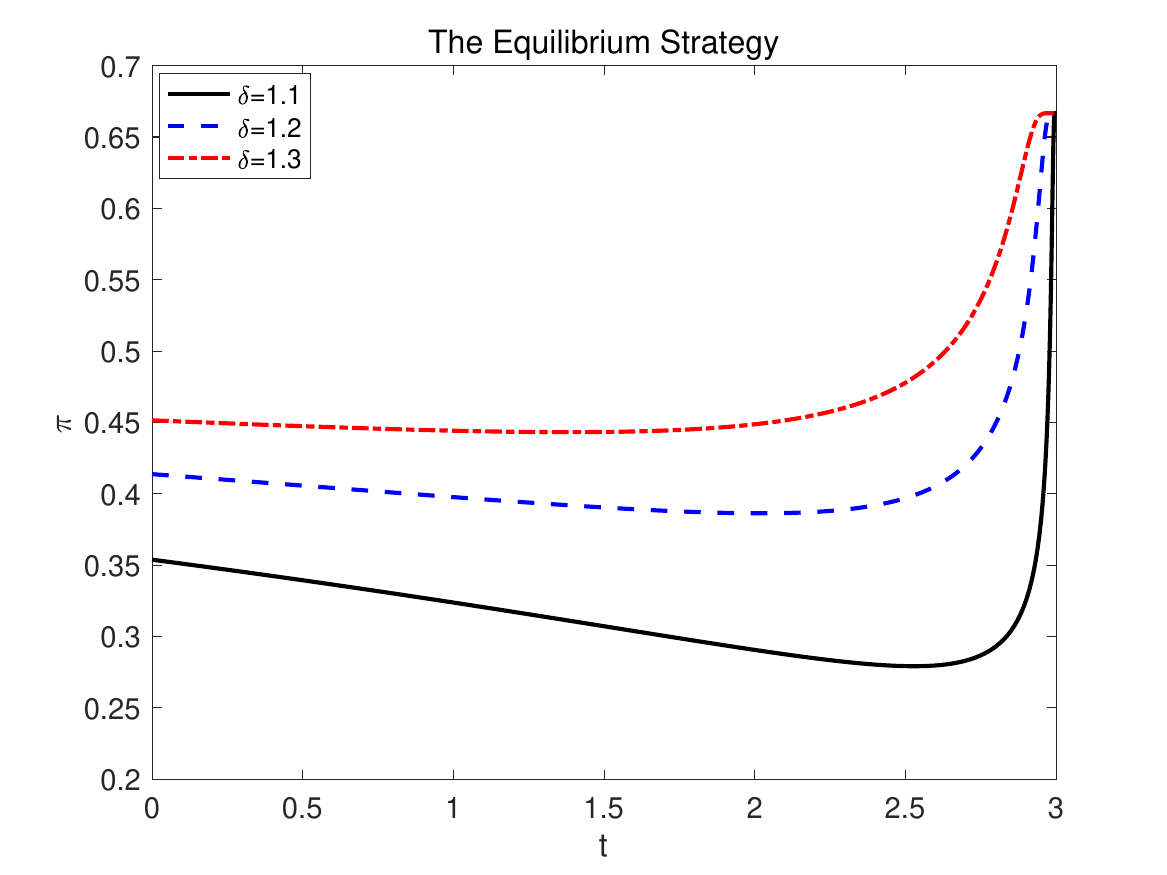}
		\caption{$\delta>1$}
	\end{subfigure}
	\caption{The equilibrium strategies for various $\delta$ with fixed $\beta=0.5$.}\label{pi_delta}
\end{figure}
\begin{figure}[ht]
	\centering
	\begin{subfigure}{0.48\textwidth}
		\centering
		\includegraphics[width=\linewidth]{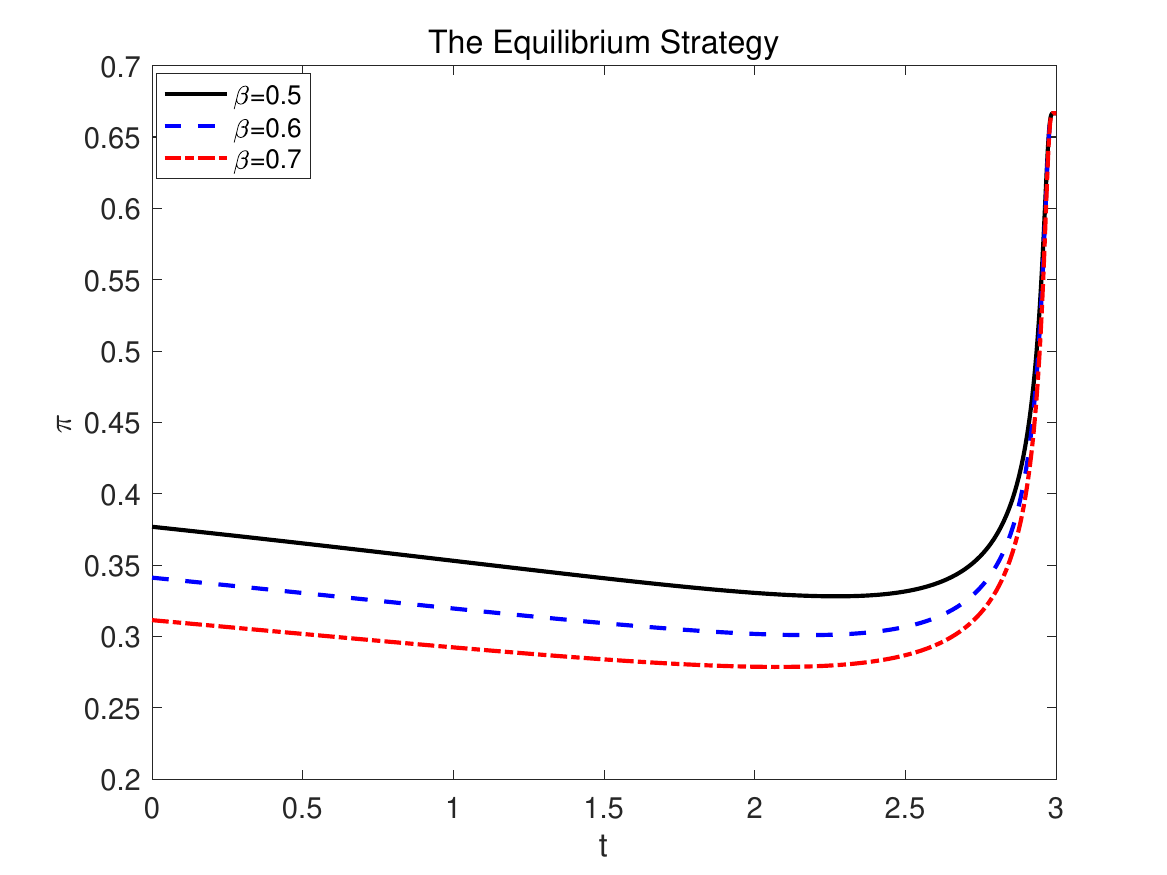}
		\caption{$\delta=0.9$}
	\end{subfigure}
	\begin{subfigure}{0.48\textwidth}
		\centering
		\includegraphics[width=\linewidth]{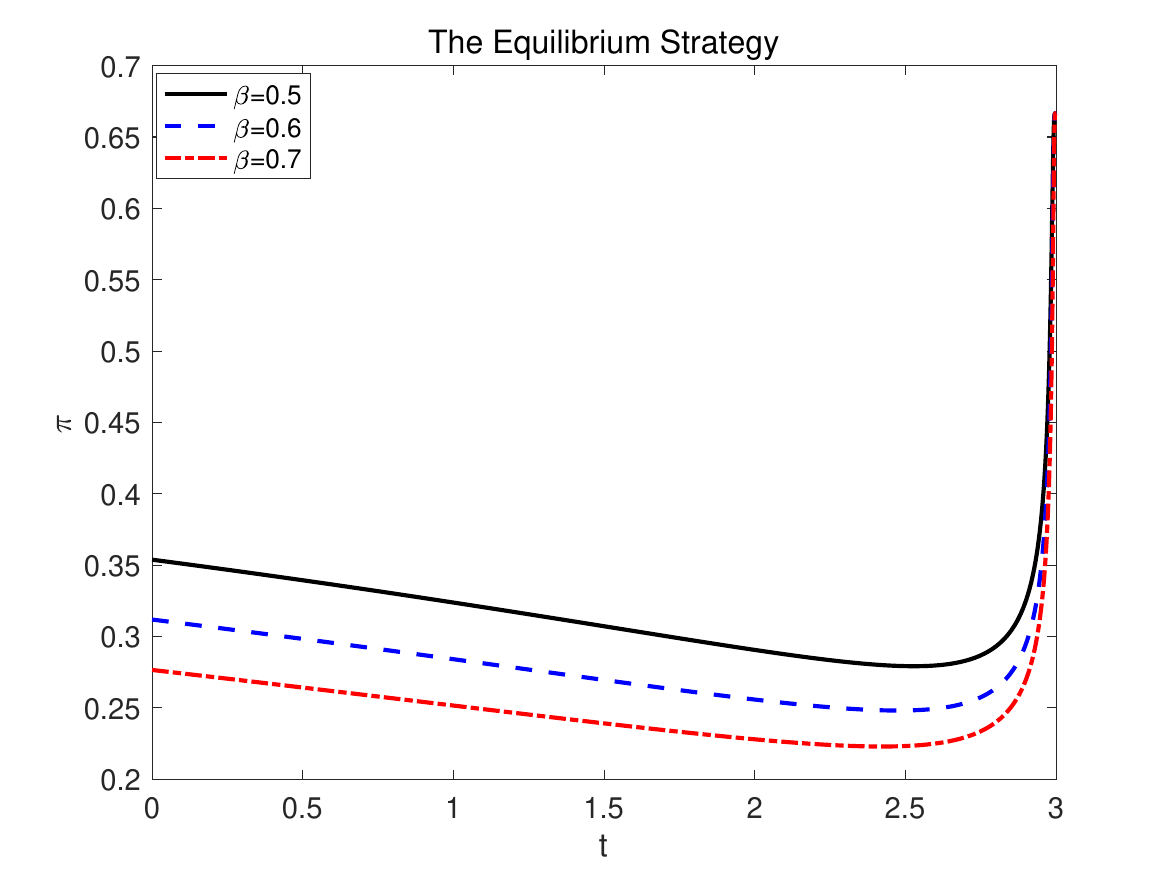}
		\caption{$\delta=1.1$}
	\end{subfigure}
	\caption{The equilibrium strategies for various $\beta$ with fixed $\delta$.}\label{pi-beta}
\end{figure}
Fig. \ref{pi_delta} displays the equilibrium strategies for fixed $\beta=0.5$ and for various $\delta$: $\delta=0.7$, $0.8$, $0.9$ in panel (a) and $\delta=1.1$, $1.2$, $1.3$ in panel (b).  The results reveal that, for fixed values of $\delta$ and $\beta$ (e.g. $\delta=0.9$ and $\beta=0.5$), the proportion invested in the risky asset gradually decreases when the time is far away from $T$, then increases and eventually converges to the Merton solution $\pi^*=\mu/(\rho \sigma^2)\approx 0.67$ when the time is near to and approaching $T$. This is  aligned with the behavior of the MRS illustrated in Section \ref{sec:char:equi}. As time $t$ approaches $T$, the cumulative risk $v(t)$ diminishes. As depicted in Figure \ref{ind_curve}, during this period, the MRS of the GDA preferences closely approximates that of the EU preference, leading to an equilibrium strategy resembling the Merton solution (see also Remark \ref{remark:converges} for a direct proof). When time $t$ deviates \emph{backwards} from $T$, the cumulative risk $v(t)$ increases, the MRS rapidly rises, and the equilibrium investment rapidly decreases from the Merton solution. When time $t$ is farther away from $T$, the cumulative risk $v(t)$ continues to increase, the MRS gradually declines, resulting in a gradual increase in investment.

In the case $\delta<1$, Fig. \ref{pi_delta}(a) shows that a larger $\delta\in(0,1)$ yields less investment in the stock. In the case $\delta>1$, however, things are reversed: Fig. \ref{pi_delta}(b) shows that a larger $\delta>1$ yields more investment in the stock.
This is reasonable: by Theorem \ref{prop:C:delta}, the GDA preference exhibits more risk aversion as $\delta$ becomes larger in $(0,1)$; in contrast, the GDA preference exhibits less risk aversion as $\delta$ becomes larger in $(1,\infty)$.

Fig. \ref{pi-beta} presents the equilibrium strategies for various values of $\beta$: $\beta=0.5$, $0.6$,  $0.7$, with fixed $\delta=0.9$ (panel (a)) or $1.1$ (panel (b)). 
The results indicate that as the investor is more averse to disappointment (larger $\beta$), the investment in the stock is smaller. This observation is in line with Theorem \ref{prop:C:delta}, which shows that the GDA preference exhibits more risk aversion as $\beta$ becomes larger.

\begin{remark}
    The  results obtained in this paper  offer useful guidance for practitioners. In particular, the characterization of equilibrium strategies under GDA preferences provides insights into how an investor’s risk attitude, shaped by parameters $\delta$ and $\beta$, affects equilibrium portfolio choices. For instance, financial advisors can interpret a lower equilibrium investment level as a signal of higher risk aversion due to greater sensitivity to cumulative risk or disappointment. Similarly, portfolio managers can use the semi-analytical strategies derived under the CRRA framework to benchmark realistic investment allocations in continuous-time settings, especially when standard discrete-time models are infeasible or provide only numerical solutions. Our findings thus bridge the gap between theoretical equilibrium analysis and practical portfolio management by illustrating how investors’ behavioral preferences systematically influence investment intensity and asset allocation.
\end{remark}

\section{A switching-drift example with a non-deterministic equilibrium}\label{sec:switching-drift}
\noindent
The preceding analysis focuses on deterministic market coefficients. In this section, we show that the same equilibrium idea also yields a genuinely non-deterministic equilibrium when the expected return is random. To keep the argument focused, we only consider the GDA preference with $\delta\neq1$ and the CRRA utility function $U_\rho$ with $\rho\in(0,1)$.

Let $B$ be a one-dimensional Brownian motion and let $\tau$ be a strictly positive random time independent of $B$. We assume  the switch is observable and work with the enlarged filtration
\[
	\mathcal{G}_t:=\mathcal{F}_t\vee\sigma(\tau\wedge t),\quad 0\le t\le T.
\]
As $\tau$ is independent of $B$, the process $B$ remains a Brownian motion under the filtration $\mathbb{G}:=\{\mathcal{G}_t\}_{0\le t\le T}$. The market coefficients are
\[
	\sigma_t\equiv\sigma>0,\qquad
	\mu_t=\mu_1\ind_{\{t<\tau\}}+\mu_2\ind_{\{\tau\le t\}},
\]
where $\mu_1,\mu_2\in\mR$. For a strategy $\pi$, the wealth process satisfies
\begin{equation}\label{eq:switch-wealth}
	\frac{\md W_t^\pi}{W_t^\pi}
	=
	\pi_t\mu_t\md t+\pi_t\sigma\md B_t,\qquad W_0^\pi=w_0>0.
\end{equation}
We use the  following notation:
\[
	\theta_t:=\sigma\pi_t,\qquad \lambda_i:=\frac{\mu_i}{\sigma},\quad i=1,2.
\]
Throughout this section we assume $\lambda_1>0$, $\lambda_2>0$ with $\lambda_1\neq \lambda_2$, and look for equilibria of the form $\pi_t=\theta_t/\sigma$, where
\begin{equation}\label{eq:switch-class}
	\theta_t=a(t)\ind_{\{t<\tau\}}+b(t)\ind_{\{\tau\le t\}},
\end{equation}
and $a,b$ are deterministic continuous functions. The definitions of admissibility, the preference functional $J(t,\pi)$, and equilibrium are the same as in Subsection \ref{subsection:equli}, with $\mathbb{G}$ in place of $\mathbb{F}$.

Let $U=U_\rho$ be the CRRA utility function in \eqref{crra:utility}, with $0<\rho<1$.
For $v\ge0$ and $y\in\mR$, define
\[
	Z(v,y):=\exp\left\{y-\frac12v+\sqrt v\,\xi\right\},
\]
where $\xi$ is a standard normal random variable; when $v=0$, we use the convention $Z(0,y)=e^y$. For $q>0$, define
\[
	\Phi(q;v,y):=\mE[U_\rho(Z(v,y))]
	-\beta\mE[(U_\rho(\delta q)-U_\rho(Z(v,y)))_+].
\]
Clearly, $\Phi$ is $C^1$ as a function of $(q,v,y)$ on the region $q>0$, $v>0$, and $y\in\mR$. Set
\[
	M_\rho(v,y):=\mE[Z(v,y)^{1-\rho}]
	=\exp\left\{(1-\rho)y-\frac12\rho(1-\rho)v\right\},
	\qquad
	d(q;v,y):=\frac{\log(\delta q)-y+\frac12v}{\sqrt v}.
\]
Then, for $q>0$, $v>0$, and $y\in\mR$,
\begin{align}
	\Phi(q;v,y)
	&=
	\frac{1}{1-\rho}\left[
	M_\rho(v,y)\left\{1+\beta N(d(q;v,y)-(1-\rho)\sqrt v)\right\}
	-\beta(\delta q)^{1-\rho} N(d(q;v,y))
	\right],\label{eq:switch-Phi}\\
	A(q;v,y)&:=\partial_y\Phi(q;v,y)
	=
	M_\rho(v,y)\left\{1+\beta N(d(q;v,y)-(1-\rho)\sqrt v)\right\},\label{eq:switch-A}\\
	B(q;v,y)&:=2\partial_v\Phi(q;v,y)
	=
	-\rho A(q;v,y)-\beta(\delta q)^{1-\rho}\frac{N'(d(q;v,y))}{\sqrt v}<0,\label{eq:switch-B}\\
	\partial_q\Phi(q;v,y)
	&=
	-\beta\delta^{1-\rho}q^{-\rho}N(d(q;v,y)).\label{eq:switch-Phiq}
\end{align}
Here $N$ and $N'$ denote the standard normal distribution function and density, respectively.
We write
\[
	C(q;v,y):=-B(q;v,y)
	=
	\rho A(q;v,y)+\beta(\delta q)^{1-\rho}\frac{N'(d(q;v,y))}{\sqrt v}.
\]
This implies $C(q;v,y)\ge \rho A(q;v,y)>0$ for all $q>0$, $v>0$, and $y\in\mR$.

\begin{remark}\label{rmk:switch-post}
	The solution structure of \eqref{eq:switch-class} can be interpreted as a type of backward induction. On $\{\tau\le t\}$, the decision maker realizes that the switch has happened, and the
    non-deterministic equilibrium problem  reduces to the deterministic problem with market price of risk $\lambda_2$, treated in Section \ref{sec:CRRA}. Therefore, she will make the post-switch decision
	\[
		b(t)=m\left(\sqrt{\mathcal{M}^{-1}\left( \lambda_2^2(T-t)\right)}\right)\lambda_2,
		\qquad 0\le t\le T,
	\]
	where $m$ and $\mathcal{M}$ are  the functions defined in Section \ref{sec:CRRA}. In particular, $b(T)=\lambda_2/\rho>0$. Taking this into consideration, she determines the pre-switch component $a$ based on $b$ and the distribution of $\tau$. 
\end{remark}

For a generic pre-switch component $a\in C([0,T])$, $t\in[0,T]$ and $r\in[t,T]$, define
\begin{align*}
	Y^a(t,r)&:=\int_t^r a(s)\lambda_1\md s+\int_r^T b(s)\lambda_2\md s,
	&
	V^a(t,r)&:=\int_t^r a(s)^2\md s+\int_r^T b(s)^2\md s.
\end{align*}
On $\{\tau\le t\}$, the conditional law of $W_T^{\bar\pi}/W_t^{\bar\pi}$ is that of
\[
	Z\left(\int_t^T b(s)^2\md s,\int_t^T b(s)\lambda_2\md s\right).
\]
If the pre-switch component is $a$, then on $\{\tau>t\}$ the only information about the switching time is that it has not occurred before $t$. For $t<T$ with $S(t):=\mP(\tau>t)>0$, define the probability measure $Q_t$ on $(t,T]$ by
\[
	Q_t(E):=\mP(\tau\wedge T\in E\mid \tau>t)
	=
	\frac{\mP(\tau\in E\cap(t,T))+\mP(\tau\ge T)\ind_{\{T\in E\}}}{S(t)},
	\quad E\in\mathcal{B}((t,T]).
\]
Then, on \(\{\tau>t\}\), for any bounded Borel function \(H\),
\[
	\mE\left[
	H\left(\frac{W_T^{\bar\pi}}{W_t^{\bar\pi}}\right)
	\middle|\mathcal{G}_t
	\right]
	=
	\int_{(t,T]} \mE[H(Z(V^a(t,r),Y^a(t,r)))]Q_t(\md r).
\]
By applying this identity to $H_n:=H\wedge n$ and then letting $n\to\infty$, the monotone convergence theorem extends it to every nonnegative Borel function $H$. In particular, it applies to $H(z)=z^{1-\rho}$.
Consequently, on $\{\tau>t\}$, the pre-switch GDA value is the unique $q^a(t)>0$ satisfying
\begin{align}\label{eq:switch-q}
	0={}&(q^a(t))^{1-\rho}
	+\beta\int_{(t,T]}
	\mE\left[((\delta q^a(t))^{1-\rho}-Z(V^a(t,r),Y^a(t,r))^{1-\rho})_+\right]Q_t(\md r)\nonumber\\
	&-\int_{(t,T]} M_\rho(V^a(t,r),Y^a(t,r))Q_t(\md r).
\end{align}
The function on the right-hand side of \eqref{eq:switch-q}, as a function of $q$, is continuous and strictly increasing on $(0,\infty)$, is negative at $0+$, and tends to $+\infty$ as $q\to\infty$. Hence $q^a(t)$ is well defined.

\begin{proposition}\label{prop:switch-spike}
	Assume that $a\in C([0,T])$ and $a(T)>0$. Fix $t<T$ with $S(t):=\mP(\tau>t)>0$. Let $q^a(t)$ solve \eqref{eq:switch-q}, and define
	\begin{equation}\label{eq:switch-AC}
	\begin{aligned}
		A^a(t)&:=\int_{(t,T]} A(q^a(t);V^a(t,r),Y^a(t,r))Q_t(\md r),\\
		C^a(t)&:=\int_{(t,T]} C(q^a(t);V^a(t,r),Y^a(t,r))Q_t(\md r).
	\end{aligned}
	\end{equation}
	Set $B^a(t):=-C^a(t)$ and
	\[
		D^a(t):=(q^a(t))^{-\rho}
		-\int_{(t,T]}\partial_q\Phi(q^a(t);V^a(t,r),Y^a(t,r))Q_t(\md r).
	\]
	Then $D^a(t)>0$, $C^a(t)>0$, and for every bounded $\mathcal{G}_t$-measurable perturbation $k$ of the portfolio strategy,
	\begin{align}
		\lim_{\epsilon\downarrow0}
		\frac{J(t,\barpi^{t,\epsilon,k})-J(t,\barpi)}{\epsilon}
		=
		\frac{\lambda_1\sigma kA^a(t)+\frac12(2a(t)\sigma k+\sigma^2k^2)B^a(t)}
		{D^a(t)}
		\quad\hbox{on }\{\tau>t\}.\label{pre-switch-derivative}
	\end{align}
	Consequently, the pre-switch equilibrium inequality at time $t$ holds for all bounded $k$ if and only if
	\begin{equation}\label{eq:switch-pre}
		a(t)=\lambda_1\frac{A^a(t)}{C^a(t)}.
	\end{equation}
\end{proposition}
\textbf{Proof.}
Fix $t<T$ and write $Q=Q_t$. We first verify that the quantities in the statement are well defined. As $b(T)>0$ and $a(T)>0$, we have $V^a(t,r)>0$ for every $r\in[t,T]$: for $r=t$ or $t<r<T$, the term $\int_r^T b(s)^2\md s$ is strictly positive, while for $r=T$, the term $\int_t^T a(s)^2\md s$ is strictly positive. As $r\mapsto V^a(t,r)$ is continuous on $[t,T]$, there is $\underline v>0$ such that $V^a(t,r)\ge\underline v$ for all $r\in(t,T]$.

Now fix a deterministic spike \(x\in\mathbb R\) in the volatility-scaled portfolio
\(\theta\) over \([t,t+\epsilon)\). The case of a bounded $\mathcal{G}_t$-measurable perturbation $k$ follows pointwise by taking $x=\sigma k(\omega)$.

For $0<\epsilon<T-t$, conditional on $\{\tau>t\}$ and $\tau\wedge T=r$, the perturbed terminal wealth ratio has lognormal law $Z(V_{\epsilon,r}^x,Y_{\epsilon,r}^x)$. The parameters are obtained by adding $x$ to the pre-switch component on $[t,t+\epsilon)$. Thus
\begin{align*}
	Y_{\epsilon,r}^x
	={}&
	\int_t^r\left(a(s)+x\ind_{[t,t+\epsilon)}(s)\right)\lambda_1\md s
	+\int_r^T\left(b(s)+x\ind_{[t,t+\epsilon)}(s)\right)\lambda_2\md s\\
	={}&
	Y^a(t,r)+x\lambda_1(\epsilon\wedge(r-t))
	+x\lambda_2(\epsilon-(r-t))_+,
\end{align*}
and
\begin{align*}
	V_{\epsilon,r}^x
	={}&
	\int_t^r\left(a(s)+x\ind_{[t,t+\epsilon)}(s)\right)^2\md s
	+\int_r^T\left(b(s)+x\ind_{[t,t+\epsilon)}(s)\right)^2\md s\\
	={}&V^a(t,r)
	+2x\int_t^{t+(\epsilon\wedge(r-t))}a(s)\md s
	+x^2(\epsilon\wedge(r-t))\\
	&+2x\ind_{\{r<t+\epsilon\}}\int_r^{t+\epsilon}b(s)\md s
	+x^2(\epsilon-(r-t))_+.
\end{align*}
Moreover, $Y_{\epsilon,r}^x\to Y^a(t,r)$ and $V_{\epsilon,r}^x\to V^a(t,r)$ uniformly in $r\in(t,T]$ as $\epsilon\downarrow0$. In particular, for all sufficiently small $\epsilon$, $V_{\epsilon,r}^x\ge\underline v/2$ uniformly in $r\in(t,T]$. Therefore the triples $(q,V_{\epsilon,r}^x,Y_{\epsilon,r}^x)$ stay in a compact subset of $(0,\infty)\times(0,\infty)\times\mR$ whenever $q$ ranges over a compact subinterval of $(0,\infty)$.

Let $q_\epsilon^x$ be the pre-switch GDA value under this perturbation. Equivalently, $q_\epsilon^x$ is the unique positive zero of
\[
	F_\epsilon^x(q)
	:=
	q^{1-\rho}-(1-\rho)\int_{(t,T]}\Phi(q;V_{\epsilon,r}^x,Y_{\epsilon,r}^x)Q(\md r).
\]
The unperturbed root is $q_0:=q^a(t)$. The preceding uniform convergence of $Y_{\epsilon,r}^x$ and $V_{\epsilon,r}^x$, together with the compactness just obtained, gives
$
	F_\epsilon^x\longrightarrow F_0
$
locally uniformly in $q\in(0,\infty)$ as $\epsilon\downarrow0$.

We next make the bounds on $q_\epsilon^x$ explicit. As $(V_{\epsilon,r}^x,Y_{\epsilon,r}^x)$ ranges over a compact subset of $(0,\infty)\times\mR$ for all small $\epsilon$ and all $r\in(t,T]$, there are constants $0<m_x<M_x<\infty$ such that
\[
	m_x
	\le
	I_\epsilon^x:=\int_{(t,T]}M_\rho(V_{\epsilon,r}^x,Y_{\epsilon,r}^x)Q(\md r)
	\le M_x.
\]
At the root $q_\epsilon^x$, the equation $F_\epsilon^x(q_\epsilon^x)=0$ is equivalent to
\[
	I_\epsilon^x
	=
	(q_\epsilon^x)^{1-\rho}
	+
	\beta\int_{(t,T]}
	\mE\left[((\delta q_\epsilon^x)^{1-\rho}-Z(V_{\epsilon,r}^x,Y_{\epsilon,r}^x)^{1-\rho})_+\right]Q(\md r).
\]
Hence
\[
	(q_\epsilon^x)^{1-\rho}\le I_\epsilon^x\le M_x,
\]
and, as $(u-v)_+\le u$ for $u,v\ge0$,
\[
	m_x\le I_\epsilon^x
	\le
	(1+\beta\delta^{1-\rho})(q_\epsilon^x)^{1-\rho}.
\]
Therefore, for all sufficiently small $\epsilon$,
\[
	0<q_-:=
	\left(\frac{m_x}{1+\beta\delta^{1-\rho}}\right)^{\frac{1}{1-\rho}}
	\le q_\epsilon^x\le
	M_x^{\frac{1}{1-\rho}}=:q_+<\infty.
\]
Let $\epsilon_n\downarrow0$. Every subsequential limit of $q_{\epsilon_n}^x$ lies in $[q_-,q_+]$. If $q_{\epsilon_{n_j}}^x\to q_*$, then the locally uniform convergence of $F_{\epsilon}^x$ gives
\[
	F_0(q_*)=\lim_{j\to\infty}F_{\epsilon_{n_j}}^x(q_{\epsilon_{n_j}}^x)=0.
\]
By the uniqueness of the positive zero of $F_0$, $q_*=q_0$. Thus $q_\epsilon^x\to q_0$ as $\epsilon\downarrow0$.

For every fixed $r\in(t,T]$,
\[
	\frac{Y_{\epsilon,r}^x-Y^a(t,r)}{\epsilon}\longrightarrow x\lambda_1,
	\qquad
	\frac{V_{\epsilon,r}^x-V^a(t,r)}{\epsilon}\longrightarrow 2xa(t)+x^2
\qquad\hbox{as }\epsilon\downarrow0.
\]
Using $\partial_y\Phi=A$ and $2\partial_v\Phi=B$, the mean-value theorem and dominated convergence yield
\begin{align*}
	&\frac{1}{\epsilon}\int_{(t,T]}
	\left[
	\Phi(q_0;V_{\epsilon,r}^x,Y_{\epsilon,r}^x)
	-\Phi(q_0;V^a(t,r),Y^a(t,r))
	\right]Q(\md r)\\
	&\qquad\longrightarrow
	x\lambda_1A^a(t)+\frac12(2xa(t)+x^2)B^a(t).
\end{align*}
Differentiating $F_0$ with respect to $q$ gives
\begin{align*}
	\partial_qF_0(q_0)
	={}&
	(1-\rho)\left[
	(q_0)^{-\rho}
	-\int_{(t,T]}\partial_q\Phi(q_0;V^a(t,r),Y^a(t,r))Q(\md r)
	\right]\\
	={}&(1-\rho)D^a(t)>0.
\end{align*}
The last inequality follows from \eqref{eq:switch-Phiq}.
We now use the two root equations $F_\epsilon^x(q_\epsilon^x)=0$ and $F_0(q_0)=0$. They give
\[
	0
	=
	F_\epsilon^x(q_\epsilon^x)-F_0(q_0)
	=
	\left[F_\epsilon^x(q_\epsilon^x)-F_\epsilon^x(q_0)\right]
	+\left[F_\epsilon^x(q_0)-F_0(q_0)\right].
\]
By the mean-value theorem, for some $\zeta_\epsilon$ between $q_\epsilon^x$ and $q_0$,
\[
	F_\epsilon^x(q_\epsilon^x)-F_\epsilon^x(q_0)
	=
	\partial_qF_\epsilon^x(\zeta_\epsilon)(q_\epsilon^x-q_0).
\]
Applying the same compactness and dominated-convergence argument to \eqref{eq:switch-Phiq} gives $\partial_qF_\epsilon^x\to\partial_qF_0$ locally uniformly near $q_0$. Also as  $q_\epsilon^x\to q_0$, we have
\[
	\partial_qF_\epsilon^x(\zeta_\epsilon)\longrightarrow \partial_qF_0(q_0)
	=(1-\rho)D^a(t).
\]
Moreover,
\[
	\frac{F_\epsilon^x(q_0)-F_0(q_0)}{\epsilon}
	\longrightarrow
	-(1-\rho)\left[x\lambda_1A^a(t)+\frac12(2xa(t)+x^2)B^a(t)\right].
\]
Dividing the preceding identity by $\epsilon$ and passing to the limit gives
\[
	\frac{q_\epsilon^x-q_0}{\epsilon}
	\longrightarrow
	\frac{x\lambda_1A^a(t)+\frac12(2xa(t)+x^2)B^a(t)}
	{D^a(t)}\ \ ( \mbox{as}\ \ \epsilon \longrightarrow 0) .
\]
The preceding argument was for an arbitrary fixed $x$. Hence, for a bounded $\mathcal{G}_t$-measurable perturbation $k$, applying the same limit pointwise with $x=\sigma k(\omega)$ gives \eqref{pre-switch-derivative}. \eqref{eq:switch-pre} is then a direct corollary by optimizing the quadratic. \hfill $\square$

To ensure the existence of an equilibrium, we now impose a condition on the law of the switching time. Assume that \(p_T:=\mathbb P(\tau\ge T)\in (0,1)\) and that the restriction of the law of \(\tau\) to \((0,T)\) is absolutely continuous with bounded continuous density \(f\). Then
\[
	S(t):=\mP(\tau>t)=p_T+\int_t^T f(r)\md r,\qquad t<T,
\]
and the preceding definition of $Q_t$ gives, for $E\in\mathcal{B}((t,T])$,
\[
	Q_t(E)
	=
	\frac{\int_{E\cap(t,T)}f(r)\md r+p_T\ind_{\{T\in E\}}}{S(t)}.
\]
Equivalently,
\begin{equation}\label{eq:switch-Qt}
	Q_t(\md r)=\frac{f(r)}{S(t)}\md r+\frac{p_T}{S(t)}\delta_T(\md r),
	\qquad r\in(t,T].
\end{equation}

\begin{theorem}\label{thm:switching-equilibrium}
	Suppose that $U=U_\rho$ with $0<\rho<1$, $\delta\neq1$, $\lambda_1>0$ and $\lambda_2>0$. If $p_T\in(0,1)$ and the law of $\tau$ has a bounded continuous density on $(0,T)$, then there exists $a\in C([0,T])$ satisfying \eqref{eq:switch-pre} with
	$  
		a(T)=\lambda_1/\rho$ and  $ 
		|a(t)|\le |\lambda_1|/\rho$ , $ 0\le t\le T.
	$ 
	The strategy $\bar\pi_t=(1/\sigma)\bar\theta_t$ with
	\[
		\bar\theta_t=a(t)\ind_{\{t<\tau\}}+b(t)\ind_{\{\tau\le t\}},
		\qquad 0\le t<T,
	\]
	is an equilibrium.
\end{theorem}
\textbf{Proof.}
By Remark \ref{rmk:switch-post}, the fixed post-switch component $b$ is the deterministic equilibrium after the switch. Moreover $0<m(x)\le1/\rho$ and $m(0)=1/\rho$,  hence $b(T)=\lambda_2/\rho>0$.

It remains to solve the pre-switch equation. Set $a_T:=\lambda_1/\rho$ and $ M:=a_T$.
Then $M>0$. As $b(T)>0$, there exists $\eta_0>0$ such that $b$ is bounded away from zero on $[T-\eta_0,T]$. We choose $0<\eta\le\eta_0$. Define
\[
	\mathcal{K}_\eta:=
	\left\{
	a\in C([0,T]): |a(t)|\le M,\ a(T)=a_T,\ a(t)\ge\frac{M}{2}
	\text{ for }t\in[T-\eta,T]
	\right\}.
\]
Then $\mathcal{K}_\eta$ is a nonempty closed convex subset of $C([0,T])$.

For $a\in\mathcal{K}_\eta$ and $t<T$, let $q^a(t)$ be the unique solution of \eqref{eq:switch-q},  and define $A^a(t)$ and $C^a(t)$ by \eqref{eq:switch-AC}.

We first obtain uniform bounds for $q^a(t)$. As $a$ and $b$ are uniformly bounded on $\mathcal{K}_\eta$, there exist constants $0\leq \bar Y,\bar V<\infty$ such that
\[
	|Y^a(t,r)|\le\bar Y,\qquad 0\le V^a(t,r)\le\bar V
\]
for all $a\in\mathcal{K}_\eta$, $t<T$, and $r\in(t,T]$. Hence $M_\rho(V^a(t,r),Y^a(t,r))$ is bounded from above and from below by two positive constants $m_1$ and $m_0$. At the root \eqref{eq:switch-q},
\[
	(q^a(t))^{1-\rho}\le m_1,
	\qquad
	m_0\le (1+\beta\delta^{1-\rho})(q^a(t))^{1-\rho}.
\]
Therefore there are constants $0<\underline q<\overline q<\infty$, independent of $a$ and $t<T$, such that
\[
	\underline q\le q^a(t)\le\overline q\quad \text{for all }t\in[0,T).
\]
Define an operator $\Gamma$ on $\mathcal{K}_\eta$ by
\[
	(\Gamma a)(t):=\lambda_1\frac{A^a(t)}{C^a(t)},\quad t<T,
	\qquad
	(\Gamma a)(T):=a_T.
\]
As $C^a(t)\ge \rho A^a(t)>0$, we have
\[
	0<(\Gamma a)(t)\le\frac{|\lambda_1|}{\rho}=M\quad \text{for all }t\in[0,T).
\]

We next record the endpoint behavior, uniformly over $a\in\mathcal{K}_\eta$. As $t\uparrow T$,
\[
	\sup_{a\in\mathcal{K}_\eta}\sup_{r\in(t,T]}
	\left(|Y^a(t,r)|+V^a(t,r)\right)\longrightarrow0.
\]
Define
\[
	F_T(q):=q^{1-\rho}+\beta((\delta q)^{1-\rho}-1)_+-1.
\]
The same scalar root-stability argument used for $q_\epsilon^x$ in the proof of Proposition \ref{prop:switch-spike} can be applied here. Indeed, the preceding uniform convergence, together with the uniform bounds $\underline q\le q^a(t)\le\overline q$, implies that the left-hand side of \eqref{eq:switch-q}, viewed as a function of $q\in[\underline q,\overline q]$, converges to $F_T(q)$ uniformly over $a\in\mathcal{K}_\eta$ as $t\uparrow T$. As $F_T$ is continuous and strictly increasing, $F_T(0+)=-1$, and $F_T(q)\to\infty$ as $q\to\infty$, it has a unique positive root $q_*$. We now claim that
\[
	\sup_{a\in\mathcal{K}_\eta}|q^a(t)-q_*|\longrightarrow0,
	\qquad t\uparrow T.
\]
If this were false, then there would exist $\varepsilon_0>0$, $t_n\uparrow T$, and $a_n\in\mathcal{K}_\eta$ such that $|q^{a_n}(t_n)-q_*|\ge\varepsilon_0$. As $q^{a_n}(t_n)\in[\underline q,\overline q]$, a subsequence converges to some $\hat q\in[\underline q,\overline q]$. The uniform convergence of the equations and the identity defining $q^{a_n}(t_n)$ yield $F_T(\hat q)=0$. Hence $\hat q=q_*$, a contradiction.
As $\delta\neq1$, we have $\delta q_*\neq1$; indeed, if $q_*=1/\delta$, then $F_T(q_*)=\delta^{\rho-1}-1\neq0$. Hence $c_*:=|\log(\delta q_*)|>0$. The uniform convergence of $q^a(t)$ to $q_*$ and the uniform convergence of $Y^a(t,r)$ and $V^a(t,r)$ to zero imply that, for all $t$ sufficiently close to $T$,
\[
	\left|\log(\delta q^a(t))-Y^a(t,r)+\frac12V^a(t,r)\right|
	\ge\frac{c_*}{2}
\]
uniformly in $a\in\mathcal{K}_\eta$ and $r\in(t,T]$. Therefore,
\[
	\left|d(q^a(t);V^a(t,r),Y^a(t,r))\right|
	\ge \frac{c_*}{2\sqrt{V^a(t,r)}}.
\]
As $N'(z)=(2\pi)^{-1/2}\exp(-z^2/2)$ and $\sup_{a\in\mathcal{K}_\eta}\sup_{r\in(t,T]}V^a(t,r)\to0$ and the exponential term dominates $v^{-1/2}$ near zero, we have
\[
\begin{aligned}
&\sup_{a\in\mathcal{K}_\eta}\sup_{r\in(t,T]}
	\frac{N'(d(q^a(t);V^a(t,r),Y^a(t,r)))}{\sqrt{V^a(t,r)}}\\
&\quad\le
	\sup_{0<v\le \sup_{a\in\mathcal{K}_\eta}\sup_{r\in(t,T]}V^a(t,r)}
	\frac{1}{\sqrt{2\pi v}}\exp\left\{-\frac{c_*^2}{8v}\right\}
	\longrightarrow0, \  t\uparrow T. 
\end{aligned}
\]
It follows from
\[
	C(q;v,y)=\rho A(q;v,y)+\beta(\delta q)^{1-\rho}\frac{N'(d(q;v,y))}{\sqrt v}
\]
that $C^a(t)/A^a(t)\to\rho$ uniformly over $a\in\mathcal{K}_\eta$. Consequently
\[
	\sup_{a\in\mathcal{K}_\eta}|(\Gamma a)(t)-a_T|\longrightarrow0,\qquad t\uparrow T.
\]
Choosing $\eta>0$ small enough, we have $(\Gamma a)(t)\ge M/2$ on $[T-\eta,T]$, hence $\Gamma a$ satisfies the endpoint sign condition in the definition of $\mathcal{K}_\eta$.

It remains to check compactness and continuity of $\Gamma$. Fix $\varepsilon>0$. On $[0,T-\varepsilon]$, the sign condition in $\mathcal{K}_\eta$ and the fact that $b$ is bounded away from zero near $T$ imply that there exists $\kappa_\varepsilon>0$ such that
\[
	V^a(t,r)\ge\kappa_\varepsilon,\qquad
	a\in\mathcal{K}_\eta,\ t\in[0,T-\varepsilon],\ r\in(t,T].
\]
Therefore all triples $(q^a(t),V^a(t,r),Y^a(t,r))$ stay in a compact subset of $(0,\infty)\times(0,\infty)\times\mR$, on which $\Phi$, $A$, and $C$ are smooth with bounded derivatives. We first show that the roots $q^a$ are equicontinuous on $[0,T-\varepsilon]$. For
\[
	F^a(t,q):=
	q^{1-\rho}-(1-\rho)\int_{(t,T]}\Phi(q;V^a(t,r),Y^a(t,r))Q_t(\md r),
\]
the representation
\[
	\int_{(t,T]}H(t,r)Q_t(\md r)
	=
	\frac{1}{S(t)}
	\left[
	\int_t^T f(r)H(t,r)\md r+p_TH(t,T)
	\right]
\]
and the bounded continuity of $f$ imply
\[
	\sup_{a\in\mathcal{K}_\eta}\sup_{q\in[\underline q,\overline q]}
	|F^a(t,q)-F^a(t',q)|\longrightarrow0
	\quad\hbox{as } |t-t'|\to0
\]
for $t,t'\in[0,T-\varepsilon]$. Moreover, by \eqref{eq:switch-Phiq}, $\Phi(q;v,y)$ is nonincreasing in $q$. Hence, for $q_2>q_1$ in $[\underline q,\overline q]$,
\[
	F^a(t,q_2)-F^a(t,q_1)
	\ge q_2^{1-\rho}-q_1^{1-\rho}
	\ge (1-\rho)\overline q^{-\rho}(q_2-q_1).
\]
As $F^a(t,q^a(t))=0$, the preceding two displays imply that $\{q^a(\cdot):a\in\mathcal{K}_\eta\}$ is equicontinuous on $[0,T-\varepsilon]$.

Using the same integral representation, and the compactness of the range of
\[
	(q^a(t),V^a(t,r),Y^a(t,r)),
\]
we then obtain equicontinuity on $[0,T-\varepsilon]$ of the two families
\[
	\{A^a(\cdot):a\in\mathcal{K}_\eta\},
	\qquad
	\{C^a(\cdot):a\in\mathcal{K}_\eta\}.
\]
In addition, $A^a$ is bounded away from zero on this interval, because $A(q;v,y)>0$ and its arguments range over a compact set. As $C^a\ge\rho A^a$, the ratios $A^a/C^a$ are also equicontinuous on $[0,T-\varepsilon]$. Hence $\{\Gamma a:a\in\mathcal{K}_\eta\}$ is equicontinuous on $[0,T-\varepsilon]$.

Combining this equicontinuity on compact subintervals of $[0,T)$ with the uniform endpoint convergence to $a_T$, we see that $\Gamma(\mathcal{K}_\eta)$ is equicontinuous on $[0,T]$. It is also uniformly bounded by $M$, and hence Arzel\`a--Ascoli implies that $\Gamma(\mathcal{K}_\eta)$ is relatively compact in $C([0,T])$.

It remains to verify continuity of $\Gamma$. Let $a_n\to a$ uniformly in $\mathcal{K}_\eta$. On each interval $[0,T-\varepsilon]$, the lower bound $V^{a_n},V^a\ge\kappa_\varepsilon$ holds, and
\[
	\sup_{0\le t\le T-\varepsilon}\sup_{r\in(t,T]}
	\left(
	|Y^{a_n}(t,r)-Y^a(t,r)|
	+
	|V^{a_n}(t,r)-V^a(t,r)|
	\right)\longrightarrow0.
\]
It follows that $F^{a_n}(t,q)\to F^a(t,q)$ uniformly for $t\in[0,T-\varepsilon]$ and $q\in[\underline q,\overline q]$. The uniform strict monotonicity in $q$ then gives $q^{a_n}\to q^a$ uniformly on $[0,T-\varepsilon]$,  hence $A^{a_n}\to A^a$ and $C^{a_n}\to C^a$ uniformly there. As $C^a$ and $C^{a_n}$ are uniformly bounded away from zero on $[0,T-\varepsilon]$, we have $\Gamma a_n\to\Gamma a$ uniformly on $[0,T-\varepsilon]$. Finally, the endpoint convergence to $a_T$ is uniform over $\mathcal{K}_\eta$; therefore the convergence extends to the whole interval $[0,T]$. Thus $\Gamma:\mathcal{K}_\eta\to\mathcal{K}_\eta$ is continuous with relatively compact image.

By Schauder's fixed point theorem, there exists $a\in\mathcal{K}_\eta$ such that $\Gamma a=a$. This is exactly \eqref{eq:switch-pre}. By Proposition \ref{prop:switch-spike}, this gives the equilibrium property on $\{\tau>t\}$. On $\{\tau\le t\}$, the choice of $b$ is the deterministic post-switch equilibrium by Section \ref{sec:CRRA}. Hence the strategy defined by \eqref{eq:switch-class} is an equilibrium under $\mathbb{G}$.
\hfill $\square$

\section{Conclusion}
\noindent
In this paper, we study the dynamic portfolio selection problem under GDA preferences in continuous time. 
We redefine the  certainty equivalent for the GDA preferences when $\delta > 1$ and explore the monotonicity of the degree of risk aversion with respect to the parameters $\beta$ and $\delta$.  Due to the implicit definition of the certainty equivalent, our problem exhibits time inconsistency.  For strategies in $\mathcal{D}$, we find that the equilibrium condition is equivalent to achieving a balance between risk and return. This leads to a class of fully nonlinear integral equations. Through investigating the existence and uniqueness of solutions to these integral equations, we establish the uniqueness and existence of equilibrium strategies  in $\mathcal{D}$.

The results reveal that non-participation in the stock market is the unique equilibrium  in $\mathcal{D}$ under DA preferences.  Moreover, 
semi-analytical solutions in the case of CRRA utility demonstrate that the equilibrium investment under GDA preferences consistently remains below the optimal investment  under EU preferences. Our numerical analysis indicates that when $0 < \delta < 1$, a gradual increase in $\delta$ leads to a gradual decrease in equilibrium investment, while the opposite holds true for the case $\delta > 1$. Furthermore, as $\beta$ gradually increases, which indicates a higher aversion to disappointment, the equilibrium investment decreases. The observed trends align with the monotonicity of the degree of risk aversion with respect to $\beta$ and $\delta$. Finally, in a
switching-drift market with a random expected return, we prove the existence of a non-deterministic
equilibrium.
\noindent

\section*{Acknowledgement}
   The authors thank the two anonymous referees and the associate editor for their valuable comments and suggestions, which improved the previous version of this paper. Zongxia Liang is supported by the National Natural Science Foundation of
China (grants 12271290 and 12671577).
 Jianming Xia acknowledges the financial support from the
National Natural Science Foundation of China (grants  12431017 and 12471447). Fengyi Yuan acknowledges the support from the start-up fund at the Chinese University of Hong Kong (Shenzhen) (grant UDF01004253).

\appendix
\section{Some technical lemmas}
The following lemma extends \citet[Lemma 1]{Stein1981}.

\begin{lemma}\label{lma:stein}
	Let $\xi\sim N(0,1)$, $a\in\mR$ and $\phi\in AC((-\infty,a))$.\footnote{$AC((-\infty, a))$ denotes all absolutely continuous functions on $(-\infty, a)$.} If $\mE[|\phi'(\xi)|\ind_{\{\xi<a\}}]<\infty$, then $\mE[\xi \phi(\xi)\ind_{\{\xi<a\}}]=\mE[\phi'(\xi)\ind_{\{\xi<a\}}]-N'(a)\phi(a)$, where $N$ is the standard normal distribution function. 
\end{lemma}
\textbf{Proof.} 
Note that $N''(z)=-zN'(z)$, without loss of generality, we assume $a>0$.
\begin{align*}
	\mE[\phi'(\xi)\ind_{\{\xi<a\}}]&=\int_{-\infty}^a\phi'(z)N'(z)\md z\\
	&=\int_{0}^{a}\phi'(z)\left\{\int_z^{\infty}v N'(v)\md v\right\}\md z-\int_{-\infty}^{ 0}\phi'(z)\left\{\int_{-\infty}^{z}v N'(v)\md v\right\}\md z\\
	&=\int_0^{\infty}v N'(v)\left\{\int_0^{a\wedge v}\phi'(z)\md z\right\}\md v-\int_{-\infty}^{0}v N'(v)\left\{\int_v^{ 0}\phi'(z)\md z\right\}\md v\\
	&=\int_{-\infty}^{a}v N'(v)\phi(v)\md v-\int_{-\infty}^{a}v N'(v)\phi(0)\md v+\int_a^{\infty}v N'(v)\left\{\int_0^{a\wedge v}\phi'(z)\md z\right\}\md v\\
	&=\mE[\xi \phi(\xi)\ind_{\{\xi<a\}}]+N'(a)\phi(a),
\end{align*}
where the third equality has used Fubini's Theorem.\hfill $\square$ 

\begin{lemma}\label{lma:ind_diff}
	Suppose that  $a:(0,\infty)\to\mR$ and  $\psi:(0,\infty)\times \mR\to\mR$ satisfy the following conditions:
	\begin{enumerate}
		\item [(1)] $a\in C^1((0,\infty)).$
		\item[(2)] $\psi\in C((0,\infty)\times \mR)$ and $\mE|\psi(x,\xi)|<\infty$ for all $x>0$.
		\item[(3)] $\psi_x$ exists and $\mE\left[\sup\limits_{x\in[x_0-\epsilon_0,x_0+\epsilon_0]}|\psi_x(x,\xi)|\right]<\infty$ for all $x_0>0$ and some $\epsilon_0=\epsilon_0(x_0)>0$.
	\end{enumerate}
	Then 
	\begin{align*}
		\frac{\md}{\md x}\mE[\psi(x,\xi)\ind_{\xi<a(x)}]=\mE[\psi_x(x,\xi)\ind_{\xi<a(x)}]+\psi(x,a(x))N'(a(x))a'(x).
	\end{align*}
\end{lemma}
\textbf{Proof.} 
Direct computation yields
\begin{align*}
	&\frac{\mE[\psi(x+h,\xi)\ind_{\xi<a(x+h)}]-\mE[\psi(x,\xi)\ind_{\xi<a(x)}]}{h}\\
	=&\frac{\int_{-\infty}^{a(x+h)}\psi(x+h,z)N'(z)\md z-\int_{-\infty}^{a(x)}\psi(x,z)N'(z)\md z}{h}\\
	=&\frac{\int_{a(x)}^{a(x+h)}\psi(x+h,z)N'(z)\md z}{h}+\frac{\int_{-\infty}^{a(x)}(\psi(x+h,z)-\psi(x,z))N'(z)\md z}{h}.
\end{align*}
Using the mean value theorem of integral, the $C^1$ property of $a$ and the continuity of $\psi$, we have
\begin{align*}
	\lim_{h\to0} \frac{\int_{a(x)}^{a(x+h)}\psi(x+h,z)N'(z)\md z}{h}&=\lim_{h\to0}\frac{(a(x+h)-a(x))\psi(x+h,\zeta(h))N'(\zeta(h))}{h}\\
	&= a'(x)\psi(x,a(x))N'(a(x)),
\end{align*}
where  $\zeta(h)$ is between $a(x)$ and $a(x+h)$. Using the dominated convergence theorem (DCT), we have 
\begin{align*}
	\lim_{h\to0}\frac{\int_{-\infty}^{a(x)}(\psi(x+h,z)-\psi(x,z))N'(z)\md z}{h}=\int_{-\infty}^{a(x)}\psi_x(x,z)N'(z)\md z=\mE[\psi_x(x,\xi)\ind_{\xi<a(x)}].
\end{align*}
Thus, the proof follows.
\hfill $\square$ 

\begin{lemma}\label{lma:C^1condition}
	Suppose that $G\in C([0,\infty)\times\mR)\cap C^1((0,\infty)\times\mR)$ and satisfies the following conditions:
	\begin{enumerate}
		\item[(a)]
		There is a continuous function $d_1$ such that \  $\lim\limits_{x\downarrow 0,y\to y_0}G_x(x,y)=d_1(y_0)\quad\forall y_0\in \mR$.
		\item[(b)]$G_y(0,y)$  is continuous and $\lim\limits_{x\downarrow 0,y\to y_0}G_y(x,y)=G_y(0,y_0)\quad\forall y_0\in \mR$.
	\end{enumerate}
	Then  $G\in  C^1([0,\infty)\times\mR)$.
\end{lemma}
\textbf{Proof.} 
Fix $y_0\in\mR$, $G(\cdot,y_0)\in C([0,\infty))\cap C^1((0,\infty)).$ Then
\begin{align*}
	\lim_{\epsilon\downarrow0}\frac{G(\epsilon,y_0)-G(0,y_0)}{\epsilon}=\lim_{\epsilon\downarrow 0}{G_x(\epsilon,y_0)}=d_1(y_0).
\end{align*}
That is, $G_x(0,y_0)$ exists. By $(a)$, $G_x\in C([0,\infty)\times\mR)$. By $(b)$, $G_y\in C([0,\infty)\times\mR)$. Therefore, $G\in  C^1([0,\infty)\times\mR)$.
\hfill $\square$

	\section{Proofs}
	\subsection{Proof of Lemma \ref{lma:eta}}
\label{sec:proof:eta1}
Let 
\begin{align*}
	f(p)=U(p)-\mE_t\left[U(Y)\right]+\beta \mE_t\left[(U(\delta p)-U(Y))_+\right],\quad p\in(0,\infty).
\end{align*}
Obviously, $f(p)$ is $\F_t$-measurable for every $p$. 
We need to look for a continuous and strictly increasing version of $f$, which is still denoted by $f$, such that, 
for almost all $\omega$, $f(\omega,p)$ is continuous and strictly increasing in $p\in(0,\infty)$.
To this end, we  
consider the regular conditional law of $Y$ with respect to $\F_t$, which is denoted by $\mP_Y^t$. By $\mE_t[|U(Y)|]<\infty$ a.s., we have $\int_0^\infty |U(y)|\mP_Y^t(\omega,\md y)<\infty$ for almost all $\omega$.
Then, there exists some  $\Omega_0\subset\Omega$ such that $\mP(\Omega_0)=1$ and, for every $\omega\in\Omega_0$, $\int_0^\infty |U(y)|\mP_Y^t(\omega,\md y)<\infty$ and
$$f(\omega,p)=U(p)-\int_0^\infty U(y)\mP_Y^t(\omega,\md y)+\beta\int_0^\infty(U(\delta p)-U(y))_+\mP_Y^t(\omega,\md y), \quad p\in(0,\infty).$$
Obviously, for every $\omega\in\Omega_0$, $f(\omega,p)$ is  strictly increasing in $p\in(0,\infty)$. Moreover, for every $\omega\in\Omega_0$, by $\int_0^\infty |U(y)|\mP_Y^t(\omega,\md y)<\infty$ and monotone convergence, $f(\omega,p)$ is continuous in $p\in(0,\infty)$ and $$\lim_{p\to+\infty}f(\omega,p)\geq U(+\infty)-\int_0^\infty U(y)\mP_Y^t(\omega,\md y)>0.$$  If $U(0+)>-\infty$, the DCT gives
\begin{equation}\label{ineq:f(p)}
	\lim_{p\to0+}f(\omega,p)=U(0+)- \int_0^\infty U(y)\mP_Y^t(\omega,\md y)<0,\quad \omega\in\Omega_0.
\end{equation}
If $U(0+)=-\infty$, then $|U(y)-U(\epsilon)|\ind_{y<\epsilon}\leq |U(y)|\ind_{y<\epsilon}$ for all sufficiently small $\epsilon>0$ and we can also get \eqref{ineq:f(p)} by the DCT. Therefore, for every $\omega\in\Omega_0$, there exists a unique $\eta(\omega)\in(0,\infty)$ such that $f(\omega,\eta(\omega))=0$. It is left to show that $\eta$ is $\F_t$-measurable. Indeed,  we have
$$\{\omega\in\Omega_0:\eta(\omega)>a\}=\{\omega\in\Omega_0: f(\omega,a)<0\}\in\F_t\quad \forall\, a\in(0,\infty),$$
which yields the desired conclusion.

	\subsection{Three auxiliary functions}\label{sec:auxi:functions}
In this subsection, we introduce three auxiliary functions that will be used. 

Let $U\in\cR_0$.
The first function $h:[0,\infty)\times\mR\to(0,\infty)$ is defined as follows. 
For any $(x,y)\in[0,\infty)\times\mR$, consider a random variable $Z\sim\logn\left(y-x^2/2,x^2\right)$, i.e., $\log Z\sim N\left(y-x^2/2,x^2\right)$. 
By Lemma  
\ref{lma:eta}, there exists a unique constant $h(x,y)>0$ such that
\begin{equation}\label{eq:g}
	U(h(x,y))=\mE\left[U(Z)\right]-\beta \mE\left[(U(\delta h(x,y))-U(Z))_+\right].
\end{equation}
Actually, $h(x,y)=g(x^2,y)$ is the  GDA value of a random variable $Z\sim \logn\left(y-x^2/2,x^2\right)$. 
In what follows, we use $h$ instead of $g$ for the sake of simplicity in mathematical notation, avoiding the extensive use of square roots. 
\vskip 5pt
The second function $H:[0,\infty)\times\mR\to\mR$ is determined by 
\begin{align}\label{def:h&H}
	\delta h(x,y)=\mathrm{e}^{y-\frac12 x^2+H(x,y)},\quad (x,y)\in[0,\infty)\times\mR.
\end{align}
Obviously, 
\begin{equation*}
	h(0,y)\begin{cases}
		=e^y,\quad&\text{if }\delta\in(0,1],\\
		=r(y)\in(\frac{\mathrm{e}^y}{\delta},\mathrm{e}^y),\quad&\text{if }\delta>1,
	\end{cases}
\end{equation*}
where $r(y)$ is the unique solution $z\in(0,\infty)$ of the following equation 
$$
U(z)=U(\mathrm{e}^y)+\beta\left(U\left(\mathrm{e}^y\right)-U(\delta z)\right),
$$
and thus
\begin{equation}\label{H(0,y)}
	H(0,y)\begin{cases}
		=\log\delta,\quad&\text{if }\delta\in(0,1],\\
		=c(y)\in(0,\log\delta),\quad&\text{if }\delta>1,
	\end{cases}
\end{equation}
where $c(y)=\log \delta+\log \bigl(r(y)/\mathrm{e}^y\bigr)$ is the unique solution $z\in(0,\infty)$ of the following equation
\begin{align}\label{eq:c(y)}
	U\left(\frac{\mathrm{e}^{z+y}}{\delta}\right)=U(\mathrm{e}^y)+\beta\left(U\left(\mathrm{e}^y\right)-U(\mathrm{e}^{z+y})\right).
\end{align}

Hereafter, $\xi$ always represents a random variable with standard normal distribution and $N$ is the  distribution function of $\xi$. The following lemma shows that $h$ and $H$ are continuous if $U$ is $0$-th-order regular.
\begin{lemma} \label{lma:g:cont}
	Suppose $U\in\cR_0$. Then $h$ and $H$ are in $C([0,\infty)\times\mR)$. 
\end{lemma}
\textbf{Proof.}
It suffices to show $h\in C([0,\infty)\times \mR)$. 
For any $(x,y,z)\in[0,\infty)\times\mR\times(0,\infty)$, let
\begin{align*}
	f(x,y,z)=U\left(z\right)-\mE\left[ U\left(\mathrm{e}^{x\xi+y-\frac{x^2}{2}}\right)\right]+\beta\mE\left[\left(U(\delta z)-U\left(\mathrm{e}^{x\xi+y-\frac{x^2}{2}}\right)\right)_+\right].
\end{align*}
Then $f$ is continuous in $[0,\infty)\times\mR\times(0,\infty)$ by the DCT and $U\in\mathcal{R}_0$. Moreover, $f$ is strictly increasing with respect to $z$. By the definition of $h$, $h$ is the unique
function such that $f(x,y,h(x,y))=0$ for all $(x,y)\in[0,\infty)\times\mR$.
Therefore, $h\in C([0,\infty)\times \mR)$.\hfill $\square$

The following lemma indicates that, assuming $U$ possesses higher-order regularity, both $h$ and $H$ are in $C^1((0,\infty)\times\mR)$.
\begin{lemma} \label{lma:g}
	Suppose $U\in\cR_1$ and $U'>0$. Then $h$ and $H$ are in $C^1((0,\infty)\times\mR)$. Moreover, 	for $(x,y)\in(0,\infty)\times\mR$, we have
	\begin{align}
		&H_x=\frac{\mE \left[ U'\left(\mathrm{e}^{x\xi+y-\frac {x^2}{2} }\right)\mathrm{e}^{x\xi+y-\frac {x^2}{2} }\left(\beta\ind_{\left\{\xi<\frac{H}{x}\right\}}+1\right)(\xi-x)\right]}{\left(\frac{1}{\delta}U'\left(\frac{\mathrm{e}^{H+y-\frac{x^2}{2}}}{\delta}\right) +\beta U'\left(\mathrm{e}^{H+y-\frac{x^2}{2}}\right) N\left(\frac{H}{x}\right)\right)\mathrm{e}^{H+y-\frac{x^2}{2}}}+x,\label{H_x}
		\\
		&H_y=\frac{\mE  \left[U'\left(\mathrm{e}^{x\xi+y-\frac {x^2}{2} }\right)\mathrm{e}^{x\xi+y-\frac {x^2}{2} }\left(\beta\ind_{\left\{\xi<\frac{H}{x}\right\}}+1\right)\right]}{\left(\frac{1}{\delta}U'\left(\frac{\mathrm{e}^{H+y-\frac{x^2}{2}}}{\delta}\right) +\beta U'\left(\mathrm{e}^{H+y-\frac{x^2}{2}}\right) N\left(\frac{H}{x}\right)\right)\mathrm{e}^{H+y-\frac{x^2}{2}}}-1.\label{H_y}
	\end{align}
\end{lemma}
\textbf{Proof.}
Indeed, for any $(x,y,z)\in(0,\infty)\times\mR\times \mR$, let
\begin{align*}
	F(x,y,z)\!=\!U\!\left(\!\frac{\mathrm{e}^{y-\frac{x^2}{2}+z}}{\delta}\!\right)\!\!-\!\mE\left[ U\!\left(\!\mathrm{e}^{x\xi+y-\frac{x^2}{2}}\!\right)\right]\!+\!\beta U\!\left(\!\mathrm{e}^{y-\frac{x^2}{2}+z}\right)\!N\left(\frac{z}{x}\!\right)
	\!-\!\beta \mE\left[ U\!\left(\!\mathrm{e}^{x\xi+y-\frac{x^2}{2}}\!\right)\!\!\ind_{\left\{\xi<\frac{z}{x}\right\}}\right].
\end{align*}
Then, by \eqref{eq:g}, $H:(0,\infty)\times\mR\to\mR$  is the unique function such that $F\left(x,y,H(x,y)\right)=0$ for all $(x,y)\in(0,\infty)\times\mR$. 
Using Lemma \ref{lma:ind_diff} and the DCT, it is not difficult to see that $F\in C^1((0,\infty)\times\mR\times\mR)$.
Moreover, ($\tilde{y}=y-x^2/2$)
\begin{align*}
	F_z(x,{y},z)&\!=\!U'\left(\frac{\mathrm{e}^{\tilde{y}+z}}{\delta}\right)\!\frac{\mathrm{e}^{\tilde{y}+z}}{\delta}\!+\!\beta U'\left(\mathrm{e}^{\tilde{y}+z}\right)\mathrm{e}^{\tilde{y}+z}N\left(\frac{z}{x}\right)\!+\!\beta U\left(\mathrm{e}^{\tilde{y}+z}\right)\!N'\left(\frac{z}{x}\right)\frac{1}{x}\!-\!\beta U\left(\mathrm{e}^{\tilde{y}+z}\right)\!N'\left(\frac{z}{x}\right)\frac{1}{x}\\
	&=U'\left(\frac{\mathrm{e}^{\tilde{y}+z}}{\delta}\right)\frac{\mathrm{e}^{\tilde{y}+z}}{\delta}\!+\!\beta U'\left(\mathrm{e}^{\tilde{y}+z}\right)\mathrm{e}^{\tilde{y}+z}N\left(\frac{z}{x}\right)>0.
\end{align*}
Then by the implicit function theorem,  $H\in C^1((0,\infty)\times \mR)$, which obviously implies $h\in C^1((0,\infty)\times \mR)$.  Now we show \eqref{H_x}, note that $H$ satisfies
\begin{align}\label{eq:H}
	U\left(\frac{\mathrm{e}^{H+y-\frac{x^2}{2}}}{\delta}\right)\!-\!\mE\left[ U\left(\mathrm{e}^{x\xi+y-\frac {x^2}{2} }\right)\right]\!+\!\beta U\left(\mathrm{e}^{H+y-\frac{x^2}{2}}\right)N\left(\frac{H}{x}\right)
	\!-\!\beta \mE \left[U\left(\mathrm{e}^{x\xi+y-\frac {x^2}{2}}\right)\ind_{\left\{\xi<\frac{H}{x}\right\}}\right]=0.
\end{align}
Using Lemma \ref{lma:ind_diff},	differentiating the above equation with respect to $x$ yields
\begin{align*}
	&U'\left(\frac{\mathrm{e}^{H+y-\frac{x^2}{2}}}{\delta}\right) \frac{\mathrm{e}^{H+y-\frac{x^2}{2}}}{\delta}(-x+H_x)-\mE \left[U'\left(\mathrm{e}^{x\xi+y-\frac {x^2}{2} }\right)\mathrm{e}^{x\xi+y-\frac {x^2}{2} }(\xi-x)\right]\\&\!+\!\beta U'\left(\mathrm{e}^{H+y-\frac{x^2}{2}}\right)\mathrm{e}^{H+y-\frac{x^2}{2}}(-x\!+\!H_x)N\left(\frac{H}{x}\right)\!-\!\beta \mE \left[ U'\left(\mathrm{e}^{x\xi+y-\frac {x^2}{2} }\right)\mathrm{e}^{x\xi+y-\frac {x^2}{2} }(\xi-x)\ind_{\left\{\xi<\frac{H}{x}\right\}}\right]=0.
\end{align*}
Thus \eqref{H_x} follows. The proof of (\ref{H_y}) is  similar.
\hfill $\square$ 

\vskip 5pt
Suppose that $U\in\cR_1$ and $U$ is concave.\footnote{If $U\in \mathcal{R}_1$ and $U$ is concave, then obviously $U'>0$.} The third function $m:(0,\infty)\times\mR\to(0,\infty)$ is defined as follows. For any $(x,y)\in(0,\infty)\times\mR$,
\begin{align}\label{def:m:using:g}
	m(x,y)=	-\frac{g_y(x^2,y)}{2g_v(x^2,y)}.
\end{align}
To well define $m$, we need to show  $g_v\neq 0$. Noting that
\begin{align*}
	\delta g(v,y)=\delta h(\sqrt{v},y)=\mathrm{e}^{y-\frac12 v+H(\sqrt{v},y)},\quad (v,y)\in[0,\infty)\times\mR,
\end{align*}
thus, for all $(v,y)\in(0,\infty)\times \mR$, we have
\begin{align}\label{g_vg_y}
	g_v(v,y)=\left(-\frac 12+\frac{H_x(\sqrt{v},y)}{2\sqrt{v}}\right)g(v,y), \quad  g_y(v,y)=\left(1+H_y(\sqrt{v},y)\right)g(v,y).
\end{align}

We now show that $x-H_x>0$, and consequently, $g_v<0$.
\begin{lemma}\label{x>H_x}
	Suppose that $U\in\cR_1$ and $U$ is concave, then $x>H_x(x,y)$ for all $(x,y)\in(0,\infty)\times\mR$.
\end{lemma}
\textbf{Proof.}
Define $\phi_1(z)=U'\left(\mathrm{e}^{xz+y-x^2/2 }\right)$, $\phi_2(z)=\mathrm{e}^{xz+y-x^2/2 }$ and $a=H/x$, then, by $N''(z)=-zN'(z)$, we have
\begin{align*}
	&\mE \left[\xi U'\left(\mathrm{e}^{x\xi+y-\frac {x^2}{2} }\right)\mathrm{e}^{x\xi+y-\frac {x^2}{2} }\ind_{\left\{\xi<\frac{H}{x}\right\}}\right]=\mE [\xi \phi_1(\xi)\phi_2(\xi)\ind_{\{\xi<a\}}]	=\int_{-\infty}^{a}z\phi_1(z)\phi_2(z)N'(z)\md z\\
	=&\int_{-\infty}^{a}zN'(z)\left\{-\int_z^a\md (\phi_1(v)\phi_2(v))+\phi_1(a)\phi_2(a)\right\}\md z\\
	=&-\phi_1(a)\phi_2(a)N'(a)-\int_{-\infty}^{a}zN'(z)\left\{\int_z^a\phi_1(v)\md \phi_2(v)+\int_z^a \phi_2(v)\md \phi_1(v)\right\}\md z\\
	=&-\phi_1(a)\phi_2(a)N'(a)\!-\!\int_{-\infty}^{a}\left\{\int_{-\infty}^v z N'(z)\md z\right\}\phi_1(v)\md \phi_2(v)\!-\!\int_{-\infty}^{a}\left\{\int_{-\infty}^v z N'(z)\md z\right\}\phi_2(v)\md \phi_1(v)\\
	=&-\phi_1(a)\phi_2(a)N'(a)+\int_{-\infty}^{a}N'(v)\phi_1(v)\md  \phi_2(v)+\int_{-\infty}^{a}N'(v)\phi_2(v)\md  \phi_1(v)\\
	=&\!- U'\!\left(\!\mathrm{e}^{H+y-\frac {x^2}{2} }\!\right)\mathrm{e}^{H+y-\frac {x^2}{2} } N'\!\left(\frac Hx\right)\!\!+\!x\mE \! \left[ U'\left(\!\mathrm{e}^{x\xi+y-\frac {x^2}{2} }\!\right)\mathrm{e}^{x\xi+y-\frac {x^2}{2} }\!\ind_{\{\xi<\frac{H}{x}\}}\right]\!+\!\int_{-\infty}^{a}\!\!\!N'(v)\phi_2(v)\md  \phi_1(v)\\
	\leq & \!- U'\!\left(\!\mathrm{e}^{H+y-\frac {x^2}{2} }\!\right)\mathrm{e}^{H+y-\frac {x^2}{2} } N'\!\left(\frac Hx\right)\!\!+\!x\mE \! \left[ U'\left(\!\mathrm{e}^{x\xi+y-\frac {x^2}{2} }\!\right)\mathrm{e}^{x\xi+y-\frac {x^2}{2} }\!\ind_{\{\xi<\frac{H}{x}\}}\right],
\end{align*}
where we have used the integration by parts formula in the fourth equality, the Fubini theorem in the fifth equality, and that $\phi_2$ is positive and $\phi_1$ is non-increasing in  the last inequality. Thus       
\begin{align*}
	\mE \left[ U'\left(\mathrm{e}^{x\xi+y-\frac {x^2}{2} }\right)\mathrm{e}^{x\xi+y-\frac {x^2}{2} }\ind_{\left\{\xi<\frac{H}{x}\right\}}(\xi-x)\right]\leq - U'\left(\mathrm{e}^{H+y-\frac {x^2}{2} }\right)\mathrm{e}^{H+y-\frac {x^2}{2} } N'\left(\frac Hx\right)<0.      
\end{align*}
Similarly
\begin{align*}
	\mE \left[ U'\left(\mathrm{e}^{x\xi+y-\frac {x^2}{2} }\right)\mathrm{e}^{x\xi+y-\frac {x^2}{2} }(\xi-x)\right]\leq 0.
\end{align*}
Therefore, $x-H_x>0$.\hfill $\square$

Then, according to \eqref{H_x}, \eqref{H_y}, \eqref{g_vg_y}, and \eqref{def:m:using:g}, we know that $m$ can be expressed in terms of $H$. We have the following lemma:
\begin{lemma}
	For all $(x,y)\in(0,\infty)\times\mR$, we have
	\begin{align}\label{defm}
		m(x,y)=\frac{x+x H_y}{x-H_x}=-\frac{x\mE  \left[U'\left(\mathrm{e}^{x\xi+y-\frac {x^2}{2} }\right)\mathrm{e}^{x\xi+y-\frac {x^2}{2} }\left(\beta\ind_{\left\{\xi<\frac{H}{x}\right\}}+1\right)\right]}{\mE \left[ U'\left(\mathrm{e}^{x\xi+y-\frac {x^2}{2} }\right)\mathrm{e}^{x\xi+y-\frac {x^2}{2} }\left(\beta\ind_{\left\{\xi<\frac{H}{x}\right\}}+1\right)(\xi-x)\right]}.
	\end{align}
\end{lemma}
\noindent
Next, we show that, under the conditions stated in Lemma \ref{lma:g}, if $\delta\neq 1$, then $h$ and $H$ are indeed $C^1$ on $[0,\infty)\times \mathbb{R}$.
\begin{lemma}\label{lma:H:C^1}
	Suppose $\delta\neq 1$, $U\in\cR_1$ and $U'>0$ , then $h$ and $H$ are in $C^1([0,\infty)\times\mR)$. 
\end{lemma}
\textbf{Proof.}
It suffices to show $H\in C^1([0,\infty)\times\mR)$.  We need to verify that Conditions $(a)$ and $(b)$ in Lemma \ref{lma:C^1condition} hold for $H$. By \eqref{H(0,y)} and $H\in C([0,\infty)\times \mR)$, $\forall y_0\in\mR$,
\begin{align*}
	\lim_{x\downarrow 0,y\to y_0} H(x,y)=\begin{cases}
		\log \delta,\quad&0<\delta<1,\\
		c(y_0)\in(0,\log \delta),\quad&\delta>1.
	\end{cases}
\end{align*}
Thus
\begin{align*}
	\lim_{x\downarrow0,y\to y_0} H_x(x,y)=0
\end{align*}
by the DCT. Therefore, Condition $(a)$ holds. Similarly, $\forall y_0\in\mR$,
\begin{align*}
	\lim_{x\downarrow0,y\to y_0} H_y(x,y)=\begin{cases}
		0,\quad &0<\delta<1,\\
		\frac{ U'(\mathrm{e}^{y_0})\mathrm{e}^{y_0}(\beta+1)}{\left[\frac{1}{\delta}U'\left(\frac{\mathrm{e}^{c(y_0)+y_0}}{\delta}\right)+\beta U'\left(\mathrm{e}^{c(y_0)+y_0}\right)\right]\mathrm{e}^{c(y_0)+y_0}}-1,\quad&\delta>1.
	\end{cases}
\end{align*}
For $0<\delta<1$, $H(0,y)=\log \delta$, thus $H_y(0,y)=0$ and Condition $(b)$ holds. For $\delta>1$,  by \eqref{eq:c(y)}, $c(\cdot)$ satisfies
\begin{align*}
	U\left(\frac{\mathrm{e}^{c(y)+y}}{\delta}\right)=U(\mathrm{e}^y)+\beta\left(U\left(\mathrm{e}^y\right)-U(\mathrm{e}^{c(y)+y})\right).  
\end{align*}
Differentiating the above equation with respect to $y$, we have
\begin{align*}
	H_y(0,y)= c'(y)=\frac{ U'(\mathrm{e}^{y})\mathrm{e}^{y}(\beta+1)}{\left[\frac{1}{\delta}U'\left(\frac{\mathrm{e}^{c(y)+y}}{\delta}\right)+\beta U'(\mathrm{e}^{c(y)+y})\right]\mathrm{e}^{c(y)+y}}-1.
\end{align*}
Therefore, Condition $(b)$ also holds.\hfill $\square$ 

Assuming additional conditions on $U$, we present an alternative expression for $H_x$.
\begin{lemma}\label{H_x-2}
	Suppose $U\in\mathcal{R}_2$ and $U'>0$. For any $(x,y)\in(0,\infty)\times\mR$, we have
	\!\!\!\!\!\!\!\!\!\!\!\!\!\!\!\!\!\!\!\!\!\!\!\!\!\!\!	\begin{align}
		\!\!H_x\!=\!\frac{ x\mE \left[U''\left(\mathrm{e}^{x\xi+y-\frac {x^2}{2} }\right)\mathrm{e}^{2\left(x\xi+y-\frac{x^2}{2}\right) }\left(\beta\ind_{\left\{\xi<\frac{H}{x}\right\}}\!+\!1\right)\right]\!-\! \beta U'\left(\mathrm{e}^{H+y-\frac {x^2}{2} }\right)\mathrm{e}^{H+y-\frac {x^2}{2} } N'\left(\frac Hx\right)}{\left(\frac{1}{\delta}U'\left(\frac{\mathrm{e}^{H+y-\frac{x^2}{2}}}{\delta}\right) +\beta U'\left(\mathrm{e}^{H+y-\frac{x^2}{2}}\right) N\left(\frac{H}{x}\right)\right)\mathrm{e}^{H+y-\frac {x^2}{2}}}\!+\!x.\label{2:H_x}
	\end{align}
	Moreover, if $\delta\neq1$, then $H_x(x,y)/x$ can be extended continuously to $(x,y)\in[0,\infty)\times\mR$.
\end{lemma}
\textbf{Proof.}
Let $\phi(\xi)=U'\left(\mathrm{e}^{x\xi+y-x^2/2 }\right)\mathrm{e}^{x\xi+y-x^2/2 }$ and $a(x)=H/x$. By Lemma \ref{lma:stein}, we have
\begin{align*}
	\mE &\left[ \xi U'\left(\mathrm{e}^{x\xi+y-\frac {x^2}{2} }\right)\mathrm{e}^{x\xi+y-\frac {x^2}{2} }\ind_{\left\{\xi<\frac{H}{x}\right\}}\right]\\
	&=-U'\left(\mathrm{e}^{H+y-\frac {x^2}{2} }\right)\mathrm{e}^{H+y-\frac {x^2}{2} } N'\left(\frac Hx\right)+x\mE \left[U'\left(\mathrm{e}^{x\xi+y-\frac {x^2}{2} }\right)\mathrm{e}^{x\xi+y-\frac {x^2}{2} }\ind_{\left\{\xi<\frac{H}{x}\right\}}\right]\\
	&+x\mE \left[U''\left(\mathrm{e}^{x\xi+y-\frac {x^2}{2} }\right)\mathrm{e}^{2\left(x\xi+y-\frac{x^2}{2}\right) }\ind_{\left\{\xi<\frac{H}{x}\right\}}\right].
\end{align*}
Thus
\begin{align*}
	&\mE  \left[U'\left(\mathrm{e}^{x\xi+y-\frac {x^2}{2} }\right)\mathrm{e}^{x\xi+y-\frac {x^2}{2} }\ind_{\left\{\xi<\frac{H}{x}\right\}}(\xi-x)\right]
	\\&=x\mE\left[ U''\left(\mathrm{e}^{x\xi+y-\frac {x^2}{2} }\right)\mathrm{e}^{2\left(x\xi+y-\frac{x^2}{2}\right) }\ind_{\left\{\xi<\frac{H}{x}\right\}}\right]-U'\left(\mathrm{e}^{H+y-\frac {x^2}{2} }\right)\mathrm{e}^{H+y-\frac {x^2}{2} } N'\left(\frac Hx\right).
\end{align*}
Similarly,
$$\mE \left[ U'\left(\mathrm{e}^{x\xi+y-\frac {x^2}{2} }\right)\mathrm{e}^{x\xi+y-\frac {x^2}{2} }(\xi-x)\right]=x\mE \left[U''\left(\mathrm{e}^{x\xi+y-\frac {x^2}{2} }\right)\mathrm{e}^{2\left(x\xi+y-\frac{x^2}{2}\right) }\right].$$
Therefore, we get (\ref{2:H_x}). Note that \begin{align*}
	\lim_{x\downarrow0,y\to y_0} \frac{H_x(x,y)}{x}=\begin{cases}
		\frac{U''(e^{y_0})e^{y_0}}{U'(e^{y_0})}+1,\quad &0<\delta<1,\\
		\frac{ U''(e^{y_0})e^{y_0}(\beta+1)}{[\frac{1}{\delta}U'(\frac{\mathrm{e}^{c(y_0)+y_0}}{\delta})+\beta U'(\mathrm{e}^{c(y_0)+y_0})]\mathrm{e}^{c(y_0)}}+1,\quad&\delta>1.
	\end{cases}
\end{align*}
Thus, $H_x(x,y)/x$ can be extended continuously to $(x,y)\in[0,\infty)\times\mR$.
\hfill $\square$ 

Assuming  $U\in\mathcal{R}_2$ and $U''\leq0$, by \eqref{2:H_x} , we have 
\begin{align}\label{2:defm}
	m(x,y)=\frac{\mE  \left[U'\left(\mathrm{e}^{x\xi+y-\frac {x^2}{2} }\right)\mathrm{e}^{x\xi+y-\frac {x^2}{2} }\left(\beta\ind_{\left\{\xi<\frac{H}{x}\right\}}+1\right)\right]}{{-\mE \left[U''\left(\mathrm{e}^{x\xi+y-\frac {x^2}{2} }\right)\mathrm{e}^{2\left(x\xi+y-\frac{x^2}{2}\right) }\left(1+\beta\ind_{\left\{\xi<\frac{H}{x}\right\}}\right)\right]+\beta U'\left(\mathrm{e}^{H+y-\frac {x^2}{2} }\right)\mathrm{e}^{H+y-\frac {x^2}{2} } \frac{N'\left(\frac Hx\right)}{x}}}.
\end{align}
Suppose further that  $\delta\neq 1$, then, using Lemmas \ref{lma:g:cont} and \eqref{H(0,y)}, for $y_0\in\mR$, we have
\begin{align*}
	\lim_{x\downarrow0, y\to y_0}m(x,y)=
	\begin{cases}
		-\frac{U'\left(\mathrm{e}^{y_0}\right)}{U''\left(\mathrm{e}^{y_0}\right)\mathrm{e}^{y_0}},\quad&\text{if }U^\ppm(e^{y_0})<0,\\ 
		\infty, \quad&\text{if }U^\ppm(e^{y_0})=0.
	\end{cases}
\end{align*}
Thus, in the case $U^\ppm<0$, $m$ can be continuously extended to $[0,\infty)\times\mR$ so that
\begin{align}\label{m(0,y)}
	m(0,y)=-\frac{U'\left(\mathrm{e}^{y}\right)}{U''\left(\mathrm{e}^{y}\right)\mathrm{e}^{y}}\quad \forall\,y\in\mR.
\end{align}
The extension is still denoted by $m$.  

\subsection{Proof of Lemma \ref{lma:gCC^1}}\label{sec:proof:gCC1}
Suppose that $\delta\neq 1$, $U\in\mathcal{R}_2$ and $U''<0$. Using
\begin{align*}
	\delta g(v,y)=\delta h(\sqrt{v},y)=\mathrm{e}^{y-\frac12 v+H(\sqrt{v},y)},\quad (v,y)\in[0,\infty)\times\mR,
\end{align*}
we have $g\in C([0,\infty)\times \mR)\cap C^1((0,\infty)\times \mR)$ as  $H\in C([0,\infty)\times \mR)\cap C^1((0,\infty)\times \mR)$ by Lemmas \ref{lma:g:cont} and  \ref{lma:g}. 
Moreover, by \eqref{def:m:using:g} and \eqref{defm}, we have
\begin{align*}
	\frac{g_y(v,y)}{g_v(v,y)}&=-2m(\sqrt{v},y)
	\\
	&=\frac{2\mE  \left[U'(Z)Z\left(1+\beta\ind_{\{Z<\delta g(v,y)\}}\right)\right]}{{\mE \left[U''(Z)Z^2\left(1+\beta\ind_{\{Z<\delta g(v,y)\}}\right)\right]-\beta U'(\delta g(v,y))\delta g(v,y) N'\left(\frac {\log(\delta g(v,y))-y+{v\over2}}{\sqrt{v}}\right)/{\sqrt{v}}}}.
\end{align*}
Finally, as $\delta\neq1$, from \eqref{g_vg_y}, Lemmas \ref{lma:H:C^1} and  \ref{H_x-2}, we easily get  $g\in C^1([0,\infty)\times \mR)$ and  $g_v<0, g_y>0$ for all $(v,y)\in[0,\infty)\times \mR$, and $g_y(0,0)/g_v(0,0)=-2m(0,0)=2U'\left(1\right)/U''\left(1\right)$ by \eqref{m(0,y)}. 

\subsection{Proof of Lemma \ref{lma:gCC^1:delta1}}\label{sec:proof:gCC2}
To prove Lemma \ref{lma:gCC^1:delta1}, we need the following lemma:
\begin{lemma}\label{delta=1}
	Suppose that $\beta>0$, $\delta=1$ and $U\in \mathcal{R}_1$ with $U'(1)>0$. Let $\{x_n\}_{n>1}$ be a sequence of positive numbers such that $x_n\downarrow 0$ as $n\to\infty$, and   $\{y_n\}_{n>1}$ be another sequence such that  $y_n/x_n\to 0$ as $n\to\infty$. Let $z_n=H(x_n,y_n)$, $n\ge1$. Then $z_n/x_n\to c^*<0$ as $n\to\infty$, where $c^*$ is the unique solution of the following equation
	\begin{align*}
		c+\beta cN(c)+\beta N'(c)=0.
	\end{align*}
\end{lemma}
\textbf{Proof.}
We first note that the equation
\begin{align*}
	c+\beta cN(c)+\beta N'(c)=0
\end{align*}
has a unique solution $c^*<0$. Indeed, the derivative of its left-hand side is $1+\beta N(c)>0$, its value at $c=0$ is $\beta N'(0)>0$, and its limit as $c\to-\infty$ is $-\infty$.

We have $z_n\to 0$ from Lemma \ref{lma:g:cont}. Recall the function $F$ defined in the proof of Lemma \ref{lma:g}. By the definition of $H$, $(x_n,y_n,z_n)$ satisfies $F(x_n,y_n,z_n)=0$.

We next show that $\{z_n/x_n\}$ is bounded. For every fixed $c\in\mR$, we have
\begin{align*}
	\frac{F(x_n,y_n,x_nc)}{x_n}
	=\mE\left[
	\frac{U\left(\mathrm{e}^{x_nc+y_n-\frac{x_n^2}{2}}\right)
	-U\left(\mathrm{e}^{x_n\xi+y_n-\frac{x_n^2}{2}}\right)}{x_n}
	\left(1+\beta\ind_{\{\xi<c\}}\right)\right].
\end{align*}
The mean-value theorem gives
\begin{align*}
	\frac{U\left(\mathrm{e}^{x_nc+y_n-\frac{x_n^2}{2}}\right)
	-U\left(\mathrm{e}^{x_n\xi+y_n-\frac{x_n^2}{2}}\right)}{x_n}
	\longrightarrow U'(1)(c-\xi)\quad \mP\text{-a.s.}
\end{align*}
Set
\begin{align*}
	\zeta^1_{n,c}
	&:=\min\left\{x_nc+y_n-\frac{x_n^2}{2},
	x_n\xi+y_n-\frac{x_n^2}{2}\right\},\,\,
	\zeta^2_{n,c}
	:=\max\left\{x_nc+y_n-\frac{x_n^2}{2},
	x_n\xi+y_n-\frac{x_n^2}{2}\right\}.
\end{align*}
By the mean-value theorem and the growth condition in $\mathcal R_1$, there exists a constant $C_c>0$, independent of $n$, such that
\begin{align*}
	&\left|
	\frac{U\left(\mathrm{e}^{x_nc+y_n-\frac{x_n^2}{2}}\right)
	-U\left(\mathrm{e}^{x_n\xi+y_n-\frac{x_n^2}{2}}\right)}{x_n}
	\left(1+\beta\ind_{\{\xi<c\}}\right)\right|\\
	&\quad\leq
	(1+\beta)\sup_{\zeta\in[\zeta^1_{n,c},\zeta^2_{n,c}]}
	\left|U'(\mathrm{e}^{\zeta})\mathrm{e}^{\zeta}\right||c-\xi|\\
	&\quad\leq C_c\mathrm{e}^{C_c(1+|\xi|)}(1+|\xi|).
\end{align*}
The last random variable is integrable because $\xi$ is standard normal. Therefore, the DCT gives
\begin{align*}
	\lim_{n\to\infty}\frac{F(x_n,y_n,x_nc)}{x_n}
	=U'(1)\left(c+\beta cN(c)+\beta N'(c)\right).
\end{align*}
Choose fixed $c_-<c^*<c_+$. Then
$c_-+\beta c_-N(c_-)+\beta N'(c_-)<0<c_++\beta c_+N(c_+)+\beta N'(c_+)$, and hence, for all sufficiently large $n$,
\begin{align*}
	F(x_n,y_n,x_nc_-)<0<F(x_n,y_n,x_nc_+).
\end{align*}
Since $F(x_n,y_n,\cdot)$ is strictly increasing and $F(x_n,y_n,z_n)=0$, it follows that
$c_-<z_n/x_n<c_+$.

Let $\{z_{n_k}/x_{n_k}\}$ be any convergent subsequence, and relabel it as $\{z_n/x_n\}$ for simplicity. Suppose that $z_n/x_n\to c\in\mR$. Since $\{z_n/x_n\}$ is bounded, the same mean-value-theorem estimate as above gives the same type of integrable dominating random variable. Moreover,
$\ind_{\{\xi<z_n/x_n\}}\to\ind_{\{\xi<c\}}$ $\mP$-a.s., since $\mP(\xi=c)=0$. Therefore, by the same DCT argument as above,
\begin{align*}
	0
	&=\lim_{n\to\infty}\frac{F(x_n,y_n,z_n)}{x_n}=U'(1)\mE\left[(c-\xi)\left(1+\beta\ind_{\{\xi<c\}}\right)\right]=U'(1)\left(c+\beta cN(c)+\beta N'(c)\right).
\end{align*}
Since $U'(1)>0$, $c\in\mR$ satisfies $
	c+\beta cN(c)+\beta N'(c)=0$.
By the uniqueness proved above, $c=c^*$. Thus every convergent subsequence of the bounded sequence $\{z_n/x_n\}$ has the same limit $c^*$, and hence $z_n/x_n\to c^*$ as $n\to\infty$.
\hfill $\square$

Now we are going to prove Lemma \ref{lma:gCC^1:delta1}. 
Suppose that $\beta>0$, $\delta=1$ and $U\in \mathcal{R}_1$ with $U'(1)>0$. 
By Lemma \ref{delta=1}, we have
\begin{align*}
	&\lim\limits_{v\downarrow0,\frac{y}{\sqrt{v}}\to0}\frac{g(v,y)-g(0,0)}{\sqrt{v}}=\lim\limits_{v\downarrow0,\frac{y}{\sqrt{v}}\to0}\frac{\mathrm{e}^{y-\frac12 v+H(\sqrt{v},y)}-1}{\sqrt{v}}=\lim\limits_{v\downarrow0,\frac{y}{\sqrt{v}}\to0}\frac{y-\frac12 v+H(\sqrt{v},y)}{\sqrt{v}}=c^*.
\end{align*}
 Furthermore, if $U$ is concave, then Lemma \ref{delta=1} gives
$H(x,y)/x\to c^*$ as $x\downarrow0$ and $y/x\to0$.
Thus, by \eqref{defm} and the same DCT argument as above,
\begin{align*}
\lim\limits_{x\downarrow0,\frac{y}{x}\to0}\frac{m(x,y)}{x}
=
-\frac{U'(1)\left(1+\beta N(c^*)\right)}
{U'(1)\mE\left[\xi\left(1+\beta\ind_{\{\xi<c^*\}}\right)\right]}=
\frac{1+\beta N(c^*)}{\beta N'(c^*)}
=-\frac{1}{c^*},
\end{align*}
where we have used $\mE[\xi]=0$,
$\mE[\xi\ind_{\{\xi<c^*\}}]=-N'(c^*)$, and the equation satisfied by $c^*$. Hence
\begin{align*}
	&\lim\limits_{v\downarrow0,\frac{y}{\sqrt{v}}\to0}\frac{g_y(v,y)}{\sqrt{v}g_v(v,y)}=-2\lim\limits_{v\downarrow0,\frac{y}{\sqrt{v}}\to0}\frac{m(\sqrt{v},y)}{\sqrt{v}}=\frac{2}{c^*}.
\end{align*}

	\subsection{Proof of Lemma \ref{lma:mC1}}\label{sec:proof:mC1}
For $(x,y)\in(0,\infty)\times\mR$, define
\begin{align*}
	&m^1(x,y)=\mE \left[ U'\left(\mathrm{e}^{x\xi+y-\frac {x^2}{2} }\right)\mathrm{e}^{x\xi+y-\frac {x^2}{2} }\left(\beta\ind_{\left\{\xi<\frac{H}{x}\right\}}+1\right)\right],\\
	&m^2(x,y)=	{-\mE \left[U''\left(\mathrm{e}^{x\xi+y-\frac {x^2}{2} }\right)\mathrm{e}^{2\left(x\xi+y-\frac{x^2}{2}\right) }\left(1+\beta\ind_{\left\{\xi<\frac{H}{x}\right\}}\right)\right]\!+\!\beta U'\left(\mathrm{e}^{H+y-\frac {x^2}{2} }\right)\mathrm{e}^{H+y-\frac {x^2}{2} } \frac{N'\left(\frac Hx\right)}{x}}.
\end{align*}
Then 
\begin{align*}
	m(x,y)=\frac{m^1(x,y)}{m^2(x,y)}.
\end{align*}
Using Lemma \ref{lma:ind_diff}, we have
\begin{align*}
	m^1_x(x,y)&=\mE\left[  (\xi-x) U''\left(\mathrm{e}^{x\xi+y-\frac {x^2}{2} }\right)\mathrm{e}^{2(x\xi+y-\frac {x^2}{2} )}\left(\beta\ind_{\left\{\xi<\frac{H}{x}\right\}}+1\right)\right]\\&+\mE \left[ (\xi-x) U'\left(\mathrm{e}^{x\xi+y-\frac {x^2}{2} }\right)\mathrm{e}^{x\xi+y-\frac {x^2}{2} }\left(\beta\ind_{\left\{\xi<\frac{H}{x}\right\}}+1\right)\right]\\
	&+\beta U'\left(\mathrm{e}^{H+y-\frac {x^2}{2} }\right)\mathrm{e}^{H+y-\frac {x^2}{2} }N'\left(\frac{H}{x}\right)\frac{xH_x-H}{x^2},\\
	m^1_y(x,y)&=\mE\left[  U''\left(\mathrm{e}^{x\xi+y-\frac {x^2}{2} }\right)\mathrm{e}^{2(x\xi+y-\frac {x^2}{2} )}\left(\beta\ind_{\left\{\xi<\frac{H}{x}\right\}}+1\right)\right]\\&+\mE \left[ U'\left(\mathrm{e}^{x\xi+y-\frac {x^2}{2} }\right)\mathrm{e}^{x\xi+y-\frac {x^2}{2} }\left(\beta\ind_{\left\{\xi<\frac{H}{x}\right\}}+1\right)\right]\\
	&+\beta U'\left(\mathrm{e}^{H+y-\frac {x^2}{2} }\right)\mathrm{e}^{H+y-\frac {x^2}{2} }N'\left(\frac{H}{x}\right)\frac{H_y}{x}.
\end{align*}
Similarly,
\begin{align*}
	m^2_x(x,y)&= -\mE\left[ (\xi-x)U'''\left(\mathrm{e}^{x\xi+y-\frac {x^2}{2} }\right)\mathrm{e}^{3(x\xi+y-\frac{x^2}{2}) }\left(1+\beta\ind_{\left\{\xi<\frac{H}{x}\right\}}\right)\right]\\
	&-2\mE \left[(\xi-x)U''\left(\mathrm{e}^{x\xi+y-\frac {x^2}{2} }\right)\mathrm{e}^{2(x\xi+y-\frac{x^2}{2}) }\left(1+\beta\ind_{\left\{\xi<\frac{H}{x}\right\}}\right)\right]\\&- \beta U''\left(\mathrm{e}^{H+y-\frac {x^2}{2} }\right)\mathrm{e}^{2(H+y-\frac{x^2}{2}) }N'\left(\frac{H}{x}\right)\frac{xH_x-H}{x^2} \\&+\beta (H_x-x)U''\left(\mathrm{e}^{H+y-\frac {x^2}{2} }\right)\mathrm{e}^{2(H+y-\frac {x^2}{2} )} \frac{N'\left(\frac Hx\right)}{x}\\
	&+\beta (H_x-x)U'\left(\mathrm{e}^{H+y-\frac {x^2}{2} }\right)\mathrm{e}^{H+y-\frac {x^2}{2} } \frac{N'\left(\frac Hx\right)}{x}\\
	&+\beta U'\left(\mathrm{e}^{H+y-\frac {x^2}{2} }\right)\mathrm{e}^{H+y-\frac {x^2}{2} } \left[-\frac{H}{x}N'\left(\frac Hx\right)\frac{xH_x-H}{x^3}-\frac{N'\left(\frac Hx\right)}{x^2}\right],
\end{align*}
\begin{align*}
	m^2_y(x,y)&= -\mE\left[ U'''\left(\mathrm{e}^{x\xi+y-\frac {x^2}{2} }\right)\mathrm{e}^{3(x\xi+y-\frac{x^2}{2}) }\left(1+\beta\ind_{\left\{\xi<\frac{H}{x}\right\}}\right)\right]\\
	&-2\mE \left[U''\left(\mathrm{e}^{x\xi+y-\frac {x^2}{2} }\right)\mathrm{e}^{2(x\xi+y-\frac{x^2}{2}) }\left(1+\beta\ind_{\left\{\xi<\frac{H}{x}\right\}}\right)\right]\\&- \beta U''\left(\mathrm{e}^{H+y-\frac {x^2}{2} }\right)\mathrm{e}^{2(H+y-\frac{x^2}{2}) }N'\left(\frac{H}{x}\right)\frac{H_y}{x} \\&+\beta (H_y+1)U''\left(\mathrm{e}^{H+y-\frac {x^2}{2} }\right)\mathrm{e}^{2(H+y-\frac {x^2}{2} )} \frac{N'\left(\frac Hx\right)}{x}\\
	&+\beta (H_y+1)U'\left(\mathrm{e}^{H+y-\frac {x^2}{2} }\right)\mathrm{e}^{H+y-\frac {x^2}{2} } \frac{N'\left(\frac Hx\right)}{x}\\
	&+\beta U'\left(\mathrm{e}^{H+y-\frac {x^2}{2} }\right)\mathrm{e}^{H+y-\frac {x^2}{2} } \frac{-\frac{H}{x}N'\left(\frac Hx\right)\frac{H_y}{x}}{x},
\end{align*}
where we have  used the fact that $N^\ppm(z)=-zN'(z)$. Based on Lemmas \ref{lma:H:C^1} and \ref{lma:g:cont}, we have
\begin{align*}
	\lim_{x\downarrow0,y\to y_0}m^1_x(x,y)=0=\lim_{x\downarrow0,y\to y_0}m^2_x(x,y).
\end{align*}
Moreover
\begin{align*}
	&\lim_{x\downarrow0,y\to y_0}m^1_y(x,y)=\begin{cases}
		U''(\mathrm{e}^{y_0})\mathrm{e}^{2y_0}+ U'(\mathrm{e}^{y_0})\mathrm{e}^{y_0},\quad &0<\delta<1,\\
		(\beta+1)\left(U''(\mathrm{e}^{y_0})\mathrm{e}^{2y_0}+ U'(\mathrm{e}^{y_0})\mathrm{e}^{y_0}\right),\quad &\delta>1.
	\end{cases}\\
	&\lim_{x\downarrow0,y\to y_0}m^2_y(x,y)=\begin{cases}
		-U'''(\mathrm{e}^{y_0})\mathrm{e}^{3y_0}- 2U''(\mathrm{e}^{y_0})\mathrm{e}^{2y_0},\quad &0<\delta<1,\\
		-(\beta+1)\left(U'''(\mathrm{e}^{y_0})\mathrm{e}^{3y_0}+2 U''(\mathrm{e}^{y_0})\mathrm{e}^{2y_0}\right),\quad &\delta>1.
	\end{cases}
\end{align*}
Using \eqref{m(0,y)} and 
\begin{align*}
	m_x(x,y)\!=\!\frac{	m^1_x(x,y)m^2(x,y) \!-\!m^1(x,y)m^2_x(x,y)}{(m^2(x,y))^2}, \quad 	m_y(x,y)\!=\!\frac{	m^1_y(x,y)m^2(x,y) \!-\!m^1(x,y)m^2_y(x,y)}{(m^2(x,y))^2},
\end{align*}
we have
\begin{align*}
	\lim_{x\downarrow0,y\to y_0}m_x(x,y)=0 \,\,\text{and}
	\lim_{x\downarrow0,y\to y_0}m_y(x,y)=m_{y}(0,y_0).
\end{align*}
Thus,  by Lemma \ref{lma:C^1condition}, the conclusion follows.

\bibliographystyle{abbrvnat} 
\bibliography{sample}

\renewcommand{\theequation}{\thesection.\arabic{equation}}

\end{document}